%% file: RPDE_VOL_final.tex
\documentclass{article}
\usepackage[english]{babel}
\usepackage{amsmath,amssymb,xcolor,latexsym,theorem}
\usepackage{graphicx}
\usepackage{subfig}
\usepackage{geometry}

\usepackage[textsize=small]{todonotes}

\newcommand{\assign}{:=}

\newcommand{\infixand}{\text{ and }}
\newcommand{\infixor}{\text{ or }}
\newcommand{\vertiii}[1]{{\left\vert\kern-0.25ex\left\vert\kern-0.25ex\left\vert #1 
    \right\vert\kern-0.25ex\right\vert\kern-0.25ex\right\vert}}
\newcommand{\nobracket}{}

\newcommand{\tmmathbf}[1]{\ensuremath{\boldsymbol{#1}}}
\newcommand{\tmop}[1]{\ensuremath{\operatorname{#1}}}
\newcommand{\tmtextbf}[1]{\text{{\bfseries{#1}}}}
\newcommand{\tmtextit}[1]{\text{{\itshape{#1}}}}

\newenvironment{proof}{\noindent\textbf{Proof\ }}{\hspace*{\fill}$\Box$\medskip}
\newtheorem{theorem}{Theorem}[section]
\newtheorem{example}[theorem]{Example}
\newtheorem{remark}[theorem]{Remark}
\newtheorem{corollary}[theorem]{Corollary}
\newtheorem{definition}[theorem]{Definition}
\newtheorem{lemma}[theorem]{Lemma}
\newtheorem{proposition}[theorem]{Proposition}
\newtheorem{Assumption}{Assumption}
\theoremstyle{break}
\usepackage{tikz}
\usepackage{comment}
\usepackage[colorlinks,citecolor=red,urlcolor=blue,bookmarks=false,hypertexnames=true]{hyperref}
\usepackage[colorlinks]{hyperref}
\usepackage[nameinlink,capitalize]{cleveref}
\usetikzlibrary{calc}
\usetikzlibrary{shapes}
\usepackage[autostyle]{csquotes}
\makeatletter

\usetikzlibrary{calc}
\usetikzlibrary{shapes}
\usepackage[autostyle]{csquotes}
\usepackage{mathrsfs}
\usepackage{authblk}
\begin{document}

\title{Rough PDEs for local stochastic volatility models}

\author[1]{Peter Bank \thanks{bank@math.tu-berlin.de.}}
\author[2]{Christian Bayer\thanks{bayerc@wias-berlin.de}}
\author[1,2]{Peter K. Friz\thanks{friz@math.tu-berlin.de}}
\author[1,2]{Luca Pelizzari\thanks{pelizzari@wias-berlin.de}}
\affil[1]{Technische Universität Berlin}
\affil[2]{Weierstrass Institute (WIAS)}


\maketitle

\begin{abstract} In this work, we introduce a novel pricing methodology in general, possibly non-Markovian local stochastic volatility (LSV) models. We observe that by conditioning the LSV dynamics on the Brownian motion that drives the volatility, one obtains a time-inhomogeneous Markov process. Using tools from rough path theory, we describe how to precisely understand the conditional LSV dynamics and reveal their Markovian nature. The latter allows us to connect the conditional dynamics to so-called rough partial differential equations (RPDEs), through a Feynman-Kac type of formula. In terms of European pricing, conditional on realizations of one Brownian motion, we can compute conditional option prices by solving the corresponding linear RPDEs, and then average over all samples to find unconditional prices. Our approach depends only minimally on the specification of the volatility, making it applicable for a wide range of classical and rough LSV models, and it establishes a PDE pricing method for non-Markovian models.  Finally, we present a first glimpse at numerical methods for RPDEs and apply them to price European options in several rough LSV models. 
\end{abstract}

\section{Introduction} 
In mathematical finance, a large class of asset-price models can be described by dynamics of the form
\begin{equation}
  X^{t,x}_t=x, \quad d X^{t,x} _s = f (s, X^{t,x} _s) v_s dW_s + g(s, X^{t,x} _s) v_s dB_s,
  \quad 0\leq t < s \leq T \label{SDE}
\end{equation}
where $W$ and $B$ are independent Brownian motions on some filtered probability space $(\Omega,(\mathcal{F}_t),P)$, $v$ is 
adapted to the smaller filtration $(\mathcal{F}^W_t)$ generated just by $W$ and where the functions $f$, $g$ are sufficiently regular for~\eqref{SDE} to have a unique strong solution for any starting point $(t,x)$. Let us also write $(\mathcal{F}^B_t)$ for the filtration generated by $B$ and assume for conciseness that $\mathcal{F}_t = \mathcal{F}^W_t \vee \mathcal{F}^B_t$.
We view $v$ as one's preferred ``backbone'' stochastic volatility (SV) model. 
Throughout this work, we make
\begin{Assumption} \label{ass1}The process
$(v_t : 0 \le t \le T)$ is $(\mathcal{F}^W_t)$-progressive with bounded sample paths.
\end{Assumption} 
This includes all $(\mathcal{F}^W_t)$-adapted, continuous volatility specifications, and thus includes virtually all classical SV models in use, Heston, Bergomi, Stein-Stein, as well as recent {\em rough volatility} variants thereof, see \cite{bayer2023rough} and references therein.

The function pair ($f,g)$ is chosen to account for correlation, leverage effects and to calibrate to the market's implied volatility surface: the classical {\em local stochastic volatility} (LSV) specification $f(t,x) \equiv \sigma (t,x) \rho,\, g(t,x) \equiv \sigma (t,x) \sqrt{1 - \rho^2}$, with {\em leverage function} $\sigma$ and correlation parameter $\rho$, is not only contained as special case, but it in fact equivalent upon allowing for local correlation $\rho = \rho(t,x)$. 
We note that LSV models enjoy huge popularity in quantitative finance, key references include \cite{lipton2002vol} and \cite{guyon2014nonlinear}, thanks to a most efficient particle calibration method. (The mathematical justification of this method is an intriguing problem, especially in a non-Markovian setting, but since this is not related to the contribution of this work we refrain from giving a more detailed discussion.)




A typical pricing problem in such models then amounts to computing an expectation like
\begin{align*}
  E[\Phi(X^{0,x})] =   E[\Phi(X^{0,x}_t: 0 \le t \le T)]
\end{align*}
for some sufficiently regular payoff functional $\Phi$ on path space $C[0,T]$, leaving aside here (for notational simplicity only) the impact of interest rates and discounting future pay-flows.  Monte Carlo simulation yields the probably most flexible approach to this problem. It fails, however, to make use of the specific structure of the dynamics~\eqref{SDE}: A key observation here is that the Brownian motion $W$ suffices to fully specify the possibly rather involved dynamics of $v$ and so, conditioning on the full evolution of $W$, one should be able to view the dynamics~\eqref{SDE} as that of a time-inhomogeneous diffusion driven by the independent Brownian motion $B$. As a result, the computation of prices such as $E[\Phi(X^{0,x})]$ should be possible by generating Monte Carlo simulations of $v$ and its integral against $W$ and then exploit the $(\mathcal{F}^B_t)$-Markov property of the conditional dynamics to efficiently compute for each such realization a sample of the conditional expectation $E[\Phi(X^{0,x})|\mathcal{F}^W_T]$; the desired unconditional expectation $E[\Phi(X^{0,x})]$ would then be obtained by averaging over these samples.

From a {\em financial-engineering and modelling point of view}, a conceptual appealing feature of this approach is that it allows one to disentangle the volatility model generating $v$ (and its $dW$-integral) from the nonlinear local vol-functions $f$, $g$ and the second Brownian motion $B$: As soon as we can sample from our preferred stochastic volatility model, we will be able to work out the corresponding conditional local-vol like prices, with all the advantages such a (conditionally) Markovian specification affords. For vanilla options payoffs $\Phi(X^{0,x})=\phi(X^{0,x}_T)$ for instance, a Feynman-Kac-like PDE description should be possible and even for some path-dependent options such as barriers the conditional pricing-problem is easier to handle in such a conditional local-vol model. Moreover, merely requiring the realizations of $v$ and its $dW$-integral as an interface, the local-vol pricer can be implemented using Markov techniques without any knowledge of the possibly highly non-Markovian volatility model it is going to be connected to.

Obviously, the heuristic view of the dynamics~\eqref{SDE} given $W$ is challenged immediately by the need to interpret the $dW$-part in its dynamics in a mathematically rigorous way: For a fixed realization of the integrator $Y_t:=\int_0^t v_sdW_s(\omega)$ one has to give precise meaning to what is called a \emph{Rough Stochastic Differential Equation} (RSDE):
\begin{equation}
  X^{t,x,\mathbf{Y}}_t=x, \quad d X^{t,x,\mathbf{Y}} _s = f (s, X^{t,x,\mathbf{Y}} _s) d\mathbf{Y}_s + g(s, X^{t,x} _s) \mathbf{v}^{\mathbf{Y}}_s dB_s,
  \quad 0\leq t < s \leq T \label{SDEY}.
\end{equation}
The present paper explains how Lyons' {\em rough path theory} (e.g. \cite{friz2020course} and references therein) can be used to obtain such an interpretation by lifting $Y$ to a rough path $\mathbf{Y}$ with bracket $[\mathbf{Y}]_t=\int_0^t (\mathbf{v}_s^{\mathbf{Y}})^2ds$. As shown by our first main result, Theorem \ref{consistencycoro}, for any fixed, deterministic $\mathbf{Y}$, solutions $ X^{t,x,\mathbf{Y}}$ to~\eqref{SDEY} turn out to be time-inhomogeneous Markov processes with respect to the filtration $(\mathcal{F}^B_t)$. The result also identifies the \emph{conditional stock price dynamics} given the full basic volatility evolution over $[0,T]$ via these RSDE solutions:
\begin{equation*}
        \mathrm{Law}(X^{t,x}|\mathcal{F}^W_T\lor \mathcal{F}^B_t) (\omega)= \left.\mathrm{Law}( X^{t,x,\mathbf{Y}})\right|_{\mathbf{Y}=(\int_0^.v_sdW_s,\int_0^.\int_0^u v_sdW_s v_u dW_u)(\omega)}.
      \end{equation*}
     
      As a consequence, option prices $E[\Phi(X^{0,x})]$ can be computed by averaging the corresponding expectations  $E[\Phi(X^{0,x,\mathbf{Y}})]$ of our RSDE solutions $X^{t,x,\mathbf{Y}}$ over Monte Carlo samples $\mathbf{Y}$ of the general $(\mathcal{F}^W_t)$- adapted volatility model given by the $v$ and $W$. Our second main result, Theorem \ref{RPDEmain},  focuses on vanilla options $\Phi(X)=\phi(X_T)$ and shows how the Markov property of our RSDE solutions $X^{t,x,\mathbf{Y}}$ allows one to characterize
\begin{align}\label{FeynmanKacintro}
u^\mathbf{Y}(t,x):=E[\phi(X^{t,x,\mathbf{Y}}_T)] 
\end{align} 
as the Feynman-Kac-like solution to some  {\em rough partial differential equation} (RPDE). After a suitable Itô/Stratonovich-type conversion, this RPDE reads 
\begin{equation}\label{RPDEintro1}
  - d_t u^{\mathbf{Y}}  =  {L_{t}}  [u^{\mathbf{Y}}] (\mathbf{v}^{\mathbf{Y}}_t)^2dt + \Gamma_t [u^{\mathbf{Y}}] d\mathbf{Y}^g_t 
  , \qquad u^{\mathbf{Y}} (T, x)  =  \phi (x) 
\end{equation}
with the spatial differential operators
\begin{align*}
  {L } _t := \frac{1}{2} g^2 (t, x)
  \partial^2_{xx}  + f_0 (t, x) \partial_x
  \quad \text{ and } \quad
  \Gamma_t := f (t, x) \partial_x 
\end{align*} 
for $f_0(t,x):= -\frac{1}{2}\partial_x f(t,x)f(t,x)$; the rough driver $\mathbf{Y}^g := (Y,Y^2/2)$ is what we call the {\em geometrification} of the rough path $\mathbf{Y}$ and the RPDE can be made rigorous by suitably smoothed out approximations. 

We should emphasize that the emergence of rough paths methods is {\em not} triggered by some {rough} volatility specification and in fact our results are novel even in the  classical situation with diffusive SV dynamics. That said, we do cover local {\em rough} stochastic volatility models, even beyond Hurst parameter $H \in (0,1/2)$, such as recent log-modulated SV models (think ``$H=0$'', cf.~\cite{bayer2021log}). This is in contrast to \cite{bayer2020regularity} where the required rough integrals are constructed using regularity structures or, equivalently, generalized rough path considerations \cite{friz2021precise, fukasawa2022partial} which would limit the volatility models $v$ one can use. The versatility of our approach is due to an important technical idea of this work: Rather than viewing the Brownian motions $W, B$ as rough integrators, 
we take a new perspective in recognizing the integrated volatility $
\int_0^. v_t dW_t$ (and its bracket) as fundamental object, bypassing all difficulties coming from the need to delicately balance the path regularity of integrands and integrators in previous approaches.

\bigskip

By way of illustration, we first consider the pure SV case with constant $f\equiv\rho$ and $g\equiv\sqrt{1-\rho^2}$, so that the dynamics~\eqref{SDE} amount to an additive stochastic volatility model in Bachelier form. In this case, the explicit solution to \eqref{RPDEintro1} is given by
 \begin{equation*}
       u^{\mathbf{Y}}(t,x) =  \int_{\mathbb{R}} \phi (y) \frac{1}{\sqrt{2 \pi(1-\rho^2) [{\mathbf{Y}}]_{t,T}}}
   \exp \left\{ - \frac{(y - x-\rho Y_{t,T})^2}{2 (1-\rho^2)[{\mathbf{Y}}]_{t,T}} 
   \right\}dy, \qquad [{\mathbf{Y}}]_{t,T} = \int_t^T  (\mathbf{v}^{\mathbf{Y}}_s)^2 ds
  \end{equation*}
  which precisely expresses the fact that the regular conditional distribution of $X^{t,x}_T$, given $\mathcal{F}^W_T$, is Gaussian with mean $\rho \int_t^T v_s dW_s$
  and variance $(1-\rho^2)\int_t^T v_r^2 d r$. Similar formulae hold for multiplicative stochastic volatility models in Black-Scholes form,
  when $f(x) \equiv\rho x$ and $g(x) \equiv\sqrt{1-\rho^2}x$, in this case the conditional distribution is log-normal, also with explicit mean and variance, known as {\em Romano--Touzi formula}, \cite{romano1997contingent}. We should also point out that our approach is not limited to $2$-factor models (i.e. built on two Brownian motions $B,W$), cf. the discussion in Section \ref{multivariatesection}.

  \bigskip

  The approach of this paper can be compared with an existing branch in the literature, related to so-called \textit{backward stochastic partial differential equations} (BSPDE), which have been studied extensively in
 \cite{ma1997adapted, ma1999linear, ma2012non, ma2015well}. Put in our context, one can consider the random field
\begin{align}
     \label{eqstochfield}
  \tilde{u} (t,x; \omega) :=E[\phi( X_T^{t,x})|\mathcal{F}_t](\omega)
\end{align}
 and view it as part of the \textit{stochastic Feynman-Kac solution} $(\tilde{u},\psi)$ to the BSPDE
 \begin{equation*} \begin{aligned}
     \label{stochFC}
    -d_t \tilde{u} (t,x)
      &= \left [\frac{v^2_t}{2}\left (g^2(t,x)+f^2(t,x)\right )\partial^2_{xx} \tilde{u} (t,x)
     + f(t,x)v_t\partial_x \psi(t, x) \right ]dt - \psi(t, x)dW_t, \\ \tilde{u} (T,x) & = \phi(x).
 \end{aligned}
\end{equation*} In the classical SV case, that is for multiplicative fields $f(x) \equiv\rho x$ and $g(x) \equiv\sqrt{1-\rho^2}x$, this has recently been studied in \cite{bayer2022pricing}. This approach has very mild requirements on the process $v_t (\omega)$, somewhat similar to our Assumption \ref{ass1}.

Compared to the BSPDE approach, the precise nature of the Markovian dynamics of the conditional process is presented in a more transparent way in the RSDE approach put forth in this paper. More precisely, the exact dependence of the solution $u$ to the BSPDE above on either the Brownian motion $W$ or on $v$ is not spelled out in full details, see, however, \cite{bonesini2023mathfrak}. On the other hand, conditioning the option price on $\mathcal{F}_t$ as in the BSPDE approach has direct financial meaning. Working with RSDEs, and then rough PDEs, requires another averaging to arrive at the same quantity, cf. Proposition
\ref{CBLemma}.

Another related pricing approach in non-Markovian frameworks comes from \textit{path-dependent PDEs} (PPDEs), with path-derivatives as in functional It\^o calculus in the sense of \cite{dupire2019functional}. The necessary extensions to treat singular Volterra dynamics, as relevant for rough volatility, were obtained in \cite{viens2019martingale}, to which we also refer for selected pointers to the vast PPDE literature. In \cite{jacquier2019deep} the authors consider LSV-type dynamics
and describe the random field $\tilde{u}(t,x,\omega)=E[\phi( X_T^{t,x})|\mathcal{F}_t](\omega)$ as the
unique solution to a terminal value PPDE problem.
In a recent paper \cite{bonesini2023rough}, similar PPDEs were used to analyse weak discretization schemes for rough volatility models.

A comparison of all these different approaches would be desirable, but this is not the purpose of this paper. (To the best of our knowledge, even a systematic comparison between BSPDEs and PPDEs, numerically and otherwise, is not available in the literature, but see \cite{bonesini2023mathfrak} for the specific case of rough volatility models.) 

\bigskip

As final contribution of this paper, we give a first glimpse at a numerical approach useful in the general setting. Specifically, in Section \ref{sec:FD_schemes} we discuss both a first and a second order scheme towards approximating solutions to the RPDE. For the analytic benchmarks described in the previous paragraph, the ensuing Monte Carlo simulations turn out to be very accurate.
For the goal of computing option prices, we have to compare our method with plain Monte Carlo simulation, i.e., the direct approximation of $E[\phi(X_T^{0,x})]$ vs.~averaging solutions $u^{\mathbf{Y}^1}(t,x)$, $u^{\mathbf{Y}^2}(t,x), \dots$ of~\eqref{RPDEintro1} for i.i.d.\ volatility model samples $\mathbf{Y}$. By construction, the variance of these samples is smaller than the variance of $\phi(X_T^{0,x})$, however, at the cost of increasing  computational time per sample. Apart from reducing the variance, taking conditional expectations also increases regularity of the payoff $\phi$, thus enabling the use of efficient deterministic quadrature methods (such as quasi Monte Carlo or sparse grids quadrature) and even multi-level Monte Carlo. 
Along the same lines, we expect that Greeks such as the Delta $\partial_x E[\phi(X_T^{0,x})]$ and the Gamma $\partial_{xx}^2 E[\phi(X_T^{0,x})]$ can now be computed by differentiating inside the expectation,  and average $\partial_{x} u^{\mathbf{Y}}(t,x)$ and $\partial_{xx}^2u^{\mathbf{Y}}(t,x)$ over samples of $\mathbf{Y}$, thanks to potentially increased regularity. We demonstrate these advantages of our partial Monte-Carlo approach for European put options in two LSV models in Section \ref{sec:options}. Similarly, our approach may be useful for computing the density of $X_T^{0,x}$.

\bigskip

The remainder of this work is organized as follows. In Section \ref{Presection} we introduce some preliminaries from the theory of rough paths. Section \ref{RPDEsection} covers the main results of this work, starting by establishing the RSDE representation of the conditional price dynamics in Section \ref{conditioningsection}, and connecting them to RPDEs in Section \ref{FeynmanKacSection} via a Feynman-Kac formula. Section \ref{pricingRPDE} demonstrates how (conditional) option prices in LSVMs can be obtained as solutions to these RPDEs, which is explicitly illustrated in Section \ref{RVolsection} for two SV examples. In Section \ref{multivariatesection} we generalize the results to multidimensional dynamics. Section \ref{Numeriksection} focuses on numerical illustrations of our findings. We introduce two finite-difference schemes for solving RPDEs in Section \ref{sec:FD_schemes}, and apply them in Section \ref{sec:options} to price European options and to compute corresponding Greeks. Finally, Appendix \ref{RPappendix} contains the technical proofs of the main results.

\bigskip

\noindent
{\bf Acknowledgements:} CB, PKF, LP acknowledge funding by the Deutsche Forschungsgemeinschaft (DFG, German Research Foundation) under Germany's Excellence Strategy – The Berlin Mathematics Research Center MATH+ (EXC-2046/1, project ID: 390685689). PB, CB, PKF also acknowledge seed support for the DFG CRC/TRR 388 ``Rough Analysis, Stochastic 
Dynamics and Related Fields''. The authors would like to thank Xue-Mei Li for pointing out an error in the earlier version of Theorem \ref{consistencycoro}. We also would like to thank an Associate Editor and two anonymous referees for their suggestions that helped to improve our paper.
\color{black}

\section{Rough preliminaries}\label{Presection}
We begin by introducing some fundamental concepts of rough path theory, which are essential to understand the techniques henceforth. The technical details will be discussed in Appendix \ref{RPappendix} for the interested reader, and for more details about rough path theory we refer to \cite{friz2020course} and \cite{friz2010multidimensional}.  Unless stated otherwise, we consider $V$ to be any Banach space. Consider the set $\Delta_{[0,T]}:= \left \{(s,t) \in [0,T]^2: 0 \leq s \leq t \leq T\right\}$, and for any path $Y:[0,T] \rightarrow V$, we make use of the two increment notations $(\delta Y)_{s,t} = Y_{s,t} := Y_t-Y_s$ for $s\leq t.$
\begin{definition}
  For $\alpha \in (1 / 3, 1 / 2]$, the pair $\tmmathbf{Y} \assign (Y,
  \mathbb{Y})$ is called $\alpha$-Hölder rough path on $V$, where $Y:[0,
  T] \longrightarrow V$ and $\mathbb{Y} : \Delta_{[0, T]} \longrightarrow V\otimes V$
  such that
  \[ \| Y \|_{\alpha} \assign \sup_{0 \leq s < t \leq T} \frac{| Y_t - Y_s
     |}{| t - s |^{\alpha}} < \infty, \quad \| \mathbb{Y} \|_{2 \alpha}
     \assign \sup_{0 \leq s < t \leq T} \frac{| \mathbb{Y }_{s, t} |}{| t - s
     |^{2\alpha}} < \infty \]
  and \tmtextit{Chen's} relation holds true:
  \begin{equation}
    \mathbb{Y}_{s, t} = \mathbb{Y}_{s, u} + \mathbb{Y}_{u, t} + Y_{s, u}
    \otimes Y_{u, t}, \quad 0\leq s \leq u \leq t \leq T .\label{CHEN}
  \end{equation}
  We denote the space of $\alpha$-H{\"o}lder rough paths by
  $\mathscr{C}^{\alpha} ([0, T], V)$.
\end{definition}

\begin{remark}\label{RSremark} 
  If $Y:[0, T] \rightarrow V$ is smooth, one can check that for the choice
  $\mathbb{Y}_{s, t} \assign \int_s^t Y_{s, r} \otimes dY_r$, where the
  integral is defined in a Riemann-Stieltjes sense, $\tmmathbf{Y} = (Y,
  \mathbb{Y})$ defines an $\alpha$-H{\"o}lder rough path, and \tmtextit{Chen's}
  relation is a direct consequence of the additivity of the integral. If $Y$
  is only $\alpha$-H{\"o}lder continuous, this integral might not have any
  meaning, but we can think of $\mathbb{Y}$ as postulating the value of this integral. The additional structure $\mathbb{Y}$ added to the path $Y$, gives rise to a notion of integration against $dY$ for a large class of integrands, extending Riemann-Stieltjes integration. This is one of the key features of rough path theory.
\end{remark} For two $\alpha$-H{\"o}lder rough paths $\tmmathbf{Y} = (Y, \mathbb{Y})$ and
$\tmmathbf{Z} = (Z, \mathbb{Z})$, we introduce the rough path distance
\begin{equation*}
  \varrho_{\alpha} (Y, \mathbb{Y} ; Z, \mathbb{Z}) \assign \| Y - Z \|_{\alpha} + \| \mathbb{Y} - \mathbb{Z} \|_{2 \alpha} .
  \label{RPmetric}
\end{equation*}
It is not difficult to see that $\mathscr{C}^{\alpha} ([0, T], V)$ together with the map $(\mathbf{Y},\mathbf{Z}) \mapsto |Y_0-Z_0| + \varrho_{\alpha}(\mathbf{Y},\mathbf{Z})$ defines a complete metric space.

\begin{definition}
  For $\alpha \in (1 / 3, 1 / 2]$, we call an element $\tmmathbf{Y} = (Y,
  \mathbb{Y}) \in \mathscr{C}^{\alpha} ([0, T], V)$ weakly geometric rough path, if
  \begin{equation}
    \tmop{Sym} (\mathbb{Y}_{s, t}) = \frac{1}{2} Y_{s, t} \otimes Y_{s, t}, \quad 0\leq s \leq t \leq T.
    \label{geometric}
  \end{equation}
  We denote the space of weakly geometric rough paths by
  $\mathscr{C}^{\alpha}_g ([0, T], V)$.\end{definition} \begin{remark}
    If $Y:[0, T] \rightarrow V$ is smooth, we can check that
    \eqref{geometric} is nothing else than integration by parts for
    $\mathbb{Y}_{s, t} \assign \int_s^t Y_{s, r} \otimes dY_r$. If the
    latter has no initial meaning, we can think of the condition
    \eqref{geometric} as imposing this important property for the postulated
    value $\mathbb{Y}$.  \end{remark} 
When discussing rough partial differential equations in Section \ref{RPDEsection}, a crucial technique involves approximating weakly geometric rough paths using rough path lifts of more regular paths. Consider $V= \mathbb{R}^d$, and Lipschitz paths $Y^{\epsilon}:[0,T] \longrightarrow \mathbb{R}^d$. As already mentioned in the remarks above, for all $\alpha \in (1/3,1/2]$ we have the canonical lift \begin{equation*}
    \mathbf{Y}^{\epsilon}= (Y^{\epsilon},\mathbb{Y}^{\epsilon}):= \left (Y^{\epsilon},\int_0^{\cdot}Y^{\epsilon}_{0,t}\otimes dY^{\epsilon}_t\right) \in \mathscr{C}^{\alpha}_g([0,T],\mathbb{R}^d),\label{smoothRP}
\end{equation*} where the integration can be understood in a Riemann-Stieltjes sense. Now using so-called \textit{geodesic-approximations}\footnote{See \cite[Chapter 5.2]{friz2010multidimensional} for instance.}, it is possible to prove the following result, see \cite[Proposition 2.8]{friz2020course} or \cite[Proposition 8.12]{friz2010multidimensional}. \begin{proposition}\label{weakconvpropo}
    Let $\alpha \in (1/3,1/2]$ and $\mathbf{Y}=(Y,\mathbb{Y})\in \mathscr{C}_g^{\alpha}([0,T],\mathbb{R}^d)$. Then there exist Lipschitz continuous paths $Y^{\epsilon}:[0,T] \longrightarrow \mathbb{R}^d$, such that \begin{equation}
        \mathbf{Y}^{\epsilon}:= \left(Y^{\epsilon},\int_0^{\cdot}Y^{\epsilon}_{0,t}\otimes dY^{\epsilon}_t\right) \longrightarrow \mathbf{Y} =(Y,\mathbb{Y}) \text{ uniformly on }[0,T]\text{ as } \epsilon \to 0,\label{uniform convergence}
    \end{equation} and we have uniform estimates\begin{equation}
         \sup_{\epsilon} \left( \Vert Y^{\epsilon}\Vert_{\alpha} + \Vert \mathbb{Y}^{\epsilon} \Vert_{2\alpha} \right) < \infty.\label{uniform estimates}
    \end{equation} 
\end{proposition}
\begin{remark}\label{weakconvdef} Motivated by the last proposition, we say $\mathbf{Y}^{\epsilon}$ converges weakly to $\mathbf{Y}$ in $\mathscr{C}^{\alpha}([0,T],\mathbb{R}^d)$, in symbols $\mathbf{Y}^{\epsilon} \rightharpoonup \mathbf{Y}$ in $\mathscr{C}^{\alpha}$, if and only if \eqref{uniform convergence} and \eqref{uniform estimates} hold true. Using an interpolation argument for $\alpha$-Hölder norms, see for instance \cite[Proposition 5.5]{friz2010multidimensional}, $\mathbf{Y}^{\epsilon} \rightharpoonup \mathbf{Y}$ in $\mathscr{C}^{\alpha}$ then implies that $\mathbf{Y}^{\epsilon} \rightarrow \mathbf{Y}$ in $\mathscr{C}^{\alpha'}$ for all $\alpha'\in (1/3,\alpha)$, that is $\varrho_{\alpha'}(\mathbf{Y}^{\epsilon},\mathbf{Y}) \rightarrow 0$ as $\epsilon \to 0.$
\end{remark}

A crucial concept in rough path theory, particularly relevant for this paper, is the notion of \textit{rough brackets}. Rough brackets play a similar role for rough path lifts as quadratic variation does in stochastic integration.\begin{definition}
  \label{rBrackets} For any $\tmmathbf{Y} \in \mathscr{C}^{\alpha} ([0, T], V)$
  with $\alpha \in (1 / 3, 1 / 2]$, we define the rough brackets as
  \begin{equation*}
    [\tmmathbf{Y}]_{t} \assign Y_{0,t} \otimes Y_{0,t} - 2 \tmop{Sym}
    (\mathbb{Y}_{0, t}) .
  \end{equation*}
\end{definition} It is not difficult to check that $t \mapsto [\mathbf{Y}]_t \in C^{2\alpha}([0,T],\mathrm{Sym}(V\otimes V))$. The main reason to introduce rough brackets is the following result, which is discussed in more detail in Appendix \ref{Lipbracketappendix}. \begin{lemma}\label{bijectionlemma}
    Let $\alpha \in (1/3,1/2]$. Then every $\alpha$-Hölder rough path $\mathbf{Y} = (Y,\mathbb{Y})\in \mathscr{C}^{\alpha}([0,T],V)$ can be identified uniquely with a pair $(\mathbf{Y}^g,[\mathbf{Y}])$, where  $\mathbf{Y}^g=(Y,\mathbb{Y}^g)$ is a weakly geometric $\alpha$-Hölder rough path with second level $\mathbb{Y}^g = \mathbb{Y}+\frac{1}{2}\delta[\mathbf{Y}]$. In particular, we have the following bijection \begin{equation*}
  \mathscr{C}^{\alpha} ([0, T], V) \longleftrightarrow \mathscr{C}^{\alpha}_g
  ([0, T], V) \oplus C_0^{2 \alpha} ([0, T], \tmop{Sym} (V \otimes V)),
  \label{Bijection}
\end{equation*} where $C^{2\alpha}_0$ denotes the space of $2\alpha$-Hölder continuous paths starting from $0$.
\end{lemma} The weakly geometric rough path $\mathbf{Y}^g$ is sometimes called \textit{geometrification} of the rough path $\mathbf{Y}$. Notice that the expression $\mathbb{Y}^g = \mathbb{Y} +
\frac{1}{2} \delta[\tmmathbf{Y}]$ is reminiscient of the It{\^o}-Stratonovich
relation in stochastic integration theory, and we will demonstrate how it generalizes this concept. \\ \\
In the remaining part of this section, we focus on the case where $V = \mathbb{R}^d$. There are certain cases where the brackets of a rough path $\mathbf{Y}$ are Lipschitz continuous, rather than only $2\alpha$-Hölder as in the general case. This motivates to introduce the space of rough paths with Lipschitz brackets. We denote by $\mathbb{S}^d$ (resp. $\mathbb{S}_+^d$) the set of symmetric (resp. positive semidefinite symmetric) matrices.
\begin{definition}\label{defLipschitzbrackets}
  Let $\alpha \in (1 / 3, 1 / 2)$ and $\mathbf{Y} \in
  \mathscr{C}^{\alpha} ([0, T], \mathbb{R}^d)$. We say $\tmmathbf{Y}$ is a rough
  path with Lipschitz brackets, in symbols $\mathbf{Y}\in \mathscr{C}^{\alpha,1}([0,T],\mathbb{R}^d)$, if $t \mapsto [\tmmathbf{Y}]_t$ is
  Lipschitz continuous. In this case, we set \begin{equation*}
      \mathbf{V}^{\mathbf{Y}}_t := \frac{d[\mathbf{Y}]_t}{dt}\in \mathbb{S}^d,
  \end{equation*} for almost every $t$. Moreover, we say that $\mathbf{Y}$ has non-decreasing Lipschitz brackets, in symbols $\mathbf{Y}\in \mathscr{C}^{\alpha,1+}([0,T],\mathbb{R}^d)$, if additionally we have $\mathbf{V}^{\mathbf{Y}}_t \in \mathbb{S}^d_+$. In this case we denoted by $\mathbf{v}^{\mathbf{Y}}$ the unique positive semidefinite square-root of $\mathbf{V}^{\mathbf{Y}}$, that is $\mathbf{v}^{\mathbf{Y}}(\mathbf{v}^{\mathbf{Y}})^T =\mathbf{V}^{\mathbf{Y}}$ and $\mathbf{v}^{\mathbf{Y}}\in \mathbb{S}^d_+$.
\end{definition} We refer the interested reader to Appendix \ref{Lipbracketappendix} for more details about the spaces of rough paths with Lipschitz brackets. \\ \\
An important class of rough paths comes from enhancing stochastic
processes. Consider a filtered probability space $(\Omega, \mathcal{F},
(\mathcal{F}_t)_{t \geq 0}, P)$ fulfilling the usual conditions. The most
basic, but probably most important examples are the following two lifts of
standard Brownian motion, see for instance \cite[Chapter 3]{friz2020course} for details.

\begin{example}
  Consider a $d$-dimensional standard Brownian motion
  $B$. It is well-known that we can define the two rough path lifts
  $\tmmathbf{B}^{\tmop{Ito}} = (B, \mathbb{B}^{\tmop{Ito}})$ and $\tmmathbf{B
  }^{\text{Strat}} = (B, \mathbb{B}^{\text{Strat}})$, where
  \[ \mathbb{B }^{\tmop{Ito}}_{s, t} \assign \int_s^t B_{s, r} dB_r,
     \quad \mathbb{B}^{\text{Strat}}_{s, t} = \int_s^t B_{s, r} \circ
     dB_r, \]
  where $dB$ denotes the It{\^o}, and $\circ dB$ the
  Stratonovich-integration. Using standard It{\^o}-calculus, one can check
  that for any $\alpha \in (1/3, 1 / 2)$ we have $\tmmathbf{B}^{\tmop{Ito}}
  (\omega) \in \mathscr{C}^{\alpha} ([0, T], \mathbb{R}^d)$ and
  $\tmmathbf{B}^{\text{Strat}} (\omega) \in \mathscr{C }^{\alpha}_g ([0, T],
  \mathbb{R}^d)$ almost surely.
\end{example}
It is not difficult to extend the last example to continuous local martingales $M$ of the form
\begin{equation*}\label{WienerMG}
  M_t = M_0 + \int_0^t \sigma_s dB_s, \quad 0 \leq t \leq T,
\end{equation*}
where $\sigma \in \mathbb{R}^{d \times d}$ is an $(\mathcal{F}^B_t)$-progressively measurable
process such that almost surely \begin{equation*}
    \Vert \sigma \Vert_{\infty;[0,T]} := \sup_{0\leq t \leq T} | \sigma_t | < \infty.
\end{equation*}
Indeed, we can define $\tmmathbf{M}^{\text{Itô}} = (M,
\mathbb{M}^{\text{Itô}})$ and $\tmmathbf{M }^{\text{Strat}} = (M,
\mathbb{M}^{\text{Strat}})$ by \begin{equation*}
    \mathbb{ M}^{\text{Itô}}_{s, t} = \int_s^t M_{s, r} \otimes dM_r, \quad \mathbb{M}^{\text{Strat}}_{s, t} = \int_s^t M_{s, r} \otimes \circ
   dM_r . 
\end{equation*}
The proof of the following result can be found in Appendix \ref{enhancedAppendix}.
\begin{proposition}\label{Mglift}  For any $\alpha \in (1 / 3, 1 / 2) \tmmathbf{},$ we have $\tmmathbf{M}^{\text{Itô}} \in \mathscr{C}^{\alpha}([0,T],\mathbb{R}^d)$ and $\tmmathbf{M}^{\text{Strat}} \in \mathscr{C}^{\alpha}_g([0,T],\mathbb{R}^d)$ almost surely.
\end{proposition}

\begin{remark}
  \label{CorresMG} Using It{\^o}-Stratonovich correction, we can notice that
  $\mathbb{M}^{\text{Itô}}_{s, t} = \mathbb{M}^{\text{Strat}}_{s, t} -
  \frac{1}{2} \delta[M]_{s, t}$, where $[M]$ denotes the matrix of covariations $[M^i, M^j]_{1 \leq
  i, j \leq d}$. From Lemma \ref{bijectionlemma} we know that we can uniquely identify $\tmmathbf{M}^{\text{Itô}}$ with a pair $(\tmmathbf{M}^g,[\tmmathbf{M}^{\text{Itô}}])$. By uniqueness it readily follows
  that $\mathbb{M}^g = \mathbb{M}^{\text{Strat}}$ and
  $[\tmmathbf{M}^{\text{Itô}}]_{s, t} = [M]_{s, t}$ almost surely. This
  motivates to use the following notation $\tmmathbf{M} \assign
  \tmmathbf{M}^{\text{Itô}}$ and $\tmmathbf{M}^g :=
  \tmmathbf{M}^{\text{Strat}}$. Finally, since $[M]_t= \int_0^t\sigma_s \sigma^T_sds$, we have almost surely that the map $t \mapsto [M]_t$ is Lipschitz continuous, and $\frac{d[M]_t}{dt}= \sigma_t\sigma_t^T$ is bounded almost surely, thus by Definition \ref{defLipschitzbrackets}, we have $\mathbf{M}\in \mathscr{C}^{\alpha,1+}([0,T],\mathbb{R}^d)$.
\end{remark}

\section{Pricing in local stochastic volatility models using RPDEs}\label{RPDEsection}
In this section, we present the main results of this paper that establish a connection between general local stochastic volatility models and rough partial differential equations (RPDEs). \subsection{Conditioning in local stochastic volatility dynamics}\label{conditioningsection} Consider two independent standard Brownian motions $W$ and $B$, an $(\mathcal{F}_t^W)$-adapted volatility process $(v_t)_{t\in[0,T]}$, such that Assumption \ref{ass1} holds true, and set $V=v^2$. For $f,g \in C_b^3([0,T]\times \mathbb{R},\mathbb{R})$ and $I_t:=\int_0^tv_sdW_s$, we are interested in the local stochastic volatility model \begin{equation}
    \begin{aligned}
    X^{t,x}_t=x, \quad d X^{t,x} _s = f (s, X^{t,x} _s) dI_s + g(s, X^{t,x} _s) v_s dB_s,
  \quad 0\leq t < s \leq T.\label{SDErep4}
    \end{aligned}
\end{equation} Thanks to \cite[Chapter 3 Theorem 7]{protter2005stochastic} it is clear that the SDE \eqref{SDErep4} has unique strong solutions for all $(t,x)\in [0,T]\times \mathbb{R}$. It should be noted that we focus on one-dimensional dynamics in this presentation for the sake of clarity. However, there is no obstacle in generalizing the approach to multivariate dynamics, as discussed in Section \ref{multivariatesection}. \\ \\ By conditioning on the Brownian motion $W$ in \eqref{SDErep4}, we wish to disintegrate the randomness arising from the $(\mathcal{F}^W_t)$-adapted pair $(I,v)$, and obtain a Markovian nature in $B$ for the conditional dynamics. To make this rigorous, we use the notion of rough paths with non-decreasing Lipschitz brackets introduced in Definition \ref{defLipschitzbrackets}. For such a rough path $\mathbf{Y}\in \mathscr{C}^{\alpha,1+}([0,T],\mathbb{R})$, with Lipschitz continuous rough brackets $t\mapsto[\mathbf{Y}]_t$, see Definition \ref{rBrackets}, we introduced the notation \begin{equation*}
    \left ([\mathbf{Y}]_t,\frac{d[\mathbf{Y}]_t}{dt},\sqrt{\frac{d[\mathbf{Y}]_t}{dt}} \right )= \left ( \int_0^t\mathbf{V}^{\mathbf{Y}}_sds,\mathbf{V}^{\mathbf{Y}}_t,\mathbf{v}^{\mathbf{Y}}_t\right ).
\end{equation*} Recall that the Itô rough path lift $\mathbf{I}=(I,\int \delta IdI)$ of the martingale $I$ constitutes a (random) rough path in $ \mathscr{C}^{\alpha,1+}([0,T],\mathbb{R})$ with brackets $[\mathbf{I}] = [I]=\int V_tdt$, see Proposition \ref{Mglift} and Remark \ref{CorresMG}. Thus we have the following consistency \begin{equation}
   \left([I]_t,\frac{d[I]_t}{dt},\sqrt{\frac{d[I]_t}{dt}} \right )=\left ( \int_0^tV_sds,V_t,|v_t| \right )=\left ( \int_0^t\mathbf{V}^{\mathbf{I}}_sds,\mathbf{V}^{\mathbf{I}}_t,\mathbf{v}^{\mathbf{I}}_t\right ).\label{consistencyIv}
\end{equation} \color{black} Replacing the pair $(I,v)$ in the SDE \eqref{SDErep4} with the deterministic pair $(\mathbf{Y},\mathbf{v}^{\mathbf{Y}})$ for some $\mathbf{Y}\in \mathscr{C}^{\alpha,1+}([0,T],\mathbb{R})$ yields the rough stochastic differential equation (RSDE)\footnote{A comprehensive theory about more general RSDEs can be found in \cite{friz2021rough}.} \begin{equation}
    X_t^{t, x,\mathbf{Y}} = x, \quad dX_s^{t,x,\mathbf{Y}} = f(s,X_s^{t,x,\mathbf{Y}})d\mathbf{Y}_s+g(s,X_s^{t,x,\mathbf{Y}})\mathbf{v}^{\mathbf{Y}}_sdB_s, \quad t< s \leq T.  \label{RSDErep4}
\end{equation}  At least formally, considering the consistency \eqref{consistencyIv}, when choosing $\mathbf{Y}=\mathbf{I}(\omega)$ one expects the solution to equation \eqref{RSDErep4} to coincide in some sense with the solution to equation \eqref{SDErep4}. In order to treat \eqref{RSDErep4} as a genuine rough differential equation (RDE), we define the martingale $M^{\mathbf{Y}}:=\int \mathbf{v}^{\mathbf{Y}}dB$ and we want to define a joint rough path lift of $(M^{\mathbf{Y}}(\omega),Y)$\footnote{The joint-lift method for RSDEs has been studied in \cite{diehl2015levy} for the case of Brownian motion, that is $v^{\mathbf{Y}}\equiv 1$ in our setting, and more recently for general càdlàg and $p$-rough paths in \cite{friz2023rough}.}.  The following definition explains how to lift $(M(\omega),Y)$ for \emph{any} $(\mathcal{F}^B_t)$-local martingale $M$.
\begin{definition}\label{jointlift} Let $\alpha \in (1/3,1/2]$ and $\mathbf{Y}=(Y,\mathbb{Y}) \in \mathscr{C}^{\alpha}([0,T],\mathbb{R})$. For any $(\mathcal{F}^B_t)$-local martingale $M$, we define the joint lift $\mathbf{Z}^{\mathbf{Y}}(\omega)= (Z^{\mathbf{Y}} (\omega), \mathbb{Z}^{\mathbf{Y}} (\omega))$ by \begin{equation}
    Z^{\mathbf{Y}}(\omega):=(M(\omega),Y), \qquad \mathbb{Z}^{\mathbf{Y}}_{s,t}(\omega) \assign \left(\begin{array}{cc}
    \int_s^t M_{s,r}dM_r & \int_s^t M_{s,r}dY_r\\
    \int_s^t Y_{s,r}dM_r & \mathbb{Y}_{s,t}
  \end{array}\right)(\omega),\label{lift}
\end{equation}  where the first entry of $\mathbb{Z}^{\mathbf{Y}}$ is the (canonical) It{\^o} rough path lift of the local
martingale $M$, see Proposition \ref{Mglift}, $\int_s^t Y_{s,r} dM_r$ is a well-defined It{\^o}
integral, and we set $\int_s^t M_{s,r} dY_r : = M_{s,t}Y_{s,t} - \int_s^tY_{s,r} dM_r$, imposing integration by parts.
\end{definition}
The following theorem demonstrates that we can use the joint lift to ensure the well-posedness of equation \eqref{RSDErep4}. Furthermore, it establishes that the unique solution is an $(\mathcal{F}_t^B)$-Markov process, and when $\mathbf{Y} = \mathbf{I}(\omega)$, the conditional distributions, given $\mathcal{F}^W_T\lor \mathcal{F}^B_t$, of the solutions to \eqref{SDErep4} and \eqref{RSDErep4} coincide. The proof of this result is discussed in Appendix \ref{RSDEAPPendix}. \begin{theorem}  \label{consistencycoro}Let $\alpha \in (1 / 3, 1 / 2]$ and assume that $f,g \in C_b^{3}([0,T]\times \mathbb{R},\mathbb{R})$. For any $\tmmathbf{Y} = (Y,\mathbb{Y}) \in
  \mathscr{C}^{\alpha,1+} ([0, T], \mathbb{R})$ and $M^{\mathbf{Y}}_t=\int_0^t\mathbf{v}^{\mathbf{Y}}_sdB_s$, the joint lift $\tmmathbf{Z}^{\mathbf{Y}}$ almost surely defines an
  $\alpha'$-H{\"o}lder rough path for any $\alpha' \in (1 / 3, \alpha)$. Moreover, there exists a unique
  solution to the rough differential equation
  \begin{equation}
    X^{t,x,\mathbf{Y}}_t=x,\quad {dX^{t, x, \tmmathbf{Y}}_s}^{} (\omega) = (g, f) (s, X_s^{t,
    x, \tmmathbf{Y}} (\omega)) d \tmmathbf{Z  }^{\mathbf{Y}}_s (\omega),\quad 0\leq t <s \leq T, \label{RDEStr}
  \end{equation}
  for almost every $\omega$, and $X^{t,x,\mathbf{Y}}$ defines  a time-inhomogeneous Markov process. Under Assumption \ref{ass1}, for the unique solution $X^{t,x}$ to \eqref{SDErep4} it holds that for a.e. $\omega$ we have \begin{equation*}
    \mathrm{Law}\left (\left. X^{t,x} \right |\mathcal{F}^W_T\lor \mathcal{F}^B_t\right)(\omega) = \mathrm{Law}\left (\left. X^{t,x,\mathbf{I}}\right |\mathcal{F}^W_T\lor \mathcal{F}^B_t\right)(\omega) = \mathrm{Law}(\left. X^{t,x,\mathbf{Y}})\right |_{\mathbf{Y}=\mathbf{I}(\omega)}. 
  \end{equation*} 
\end{theorem} 
\begin{remark}
Our treatment of ``mixed'' rough stochastic differential equations via
what is known as ``joint rough path lift'' method, is neither the only nor the
optimal way to solve such equations, as was pointed out in \cite{friz2021rough} where a more
powerful approach via stochastic sewing was developed. Our preference for the
present methods is, in part due to its less technical demands, but also due
to the fact that in typical (constant correlation!) situations $f, g$ are in
fact proportional to each other, so that, in this situation, there is no
benefit in reducing the regularity assumption of $g$ (vs. $f$), as is made
possible by \cite{friz2021rough}. That said, the setting of
\cite{friz2021rough} would be more natural for analyzing numerical
methods for RSDEs, such as mixed Euler (w.r.t. Brownian motion $B$) and Davie-Milstein (w.r.t to the rough path $\mathbf{Y}$) schemes.

\end{remark}

\subsection{Feynman-Kac representation and rough stochastic differential equations}\label{FeynmanKacSection}
The goal of this section is to exploit the Markovian nature of the unique solution $X^{t,x,\mathbf{Y}}$ to the RDE \eqref{RDEStr}, by relating it to \textit{rough partial differential equations} (RPDEs). As we will explore more detailed in the next section, thinking of $\mathbf{Y}=\mathbf{I}(\omega)$ for some fixed $\omega$, the $(\mathcal{F}^B_t)$-Markov process $X^{t,x,\mathbf{Y}}$ can be interpreted as the LSV price dynamics \eqref{SDErep4} conditioned on the Brownian motion $W$. Thus, it is desirable to characterize expected values of the form $E[\phi(X_T^{t,x,\mathbf{Y}})]$ for European payoff functions $\phi$, which for simplicity we assume to be continuous and bounded. It turns out that $(t,x) \mapsto E[\phi(X^{t,x,\mathbf{Y}}_T)]$ can be obtained from solutions to certain RPDEs, through a Feynman-Kac type of formula. \\ \\More specifically, let $\phi \in C_b^0(\mathbb{R},\mathbb{R})$ and fix $\mathbf{Y}=(Y,\mathbb{Y})\in \mathscr{C}^{\alpha,1+}([0,T],\mathbb{R})$, and recall the unique corresponding pair $(\mathbf{Y}^g,[\mathbf{Y}])$ as described in Lemma \ref{bijectionlemma}. We are interested in the linear second order RPDE
\begin{equation}
\begin{cases}
    - d_t u^{\mathbf{Y}} & =  {L_{t}}  [u^{\mathbf{Y}}] \mathbf{V}^{\mathbf{Y}}_tdt + \Gamma_t [u^{\mathbf{Y}}] d\mathbf{Y}_t^g, \quad (t,x) \in [0,T[ \times \mathbb{R}\\ u^{\mathbf{Y}} (T, x) & =  \phi (x), \quad x\in \mathbb{R}, 
    \end{cases} \label{RPDECl}
\end{equation} where
\begin{eqnarray*}
  {L } _t[u^{\mathbf{Y}}] (x) & \assign & \frac{1}{2} g^2 (t, x)
  \partial^2_{xx} u^{\mathbf{Y}} (t, x) + f_0 (t, x) \partial_x u (t, x)
  \\\
  \Gamma_t [u^{\mathbf{Y}}] (x) & \assign & f (t, x) \partial_x u^{\mathbf{Y}} (t, x),
\end{eqnarray*} and $f_0(t,x) := -\frac{1}{2}f(t,x)\partial_xf(t,x)$. Let us quickly describe how we define solutions to \eqref{RPDECl}. By Proposition \ref{weakconvpropo}, we can find $Y^{\epsilon} \in \mathrm{Lip}([0,T],\mathbb{R})$, such that its (canonical) rough path lift $\mathbf{Y}^{g,\epsilon} \rightharpoonup \mathbf{Y}^g$ in $\mathscr{C}^{\alpha}$ as $\epsilon \to 0$, see also Remark \ref{weakconvdef}. Replacing $d\mathbf{Y}^g_t$ with $dY^{\epsilon}_t= \dot{Y}^{\epsilon}_tdt$ in the RPDE \eqref{RPDECl}, we obtain a classical (backward) PDE
of the form \begin{equation}
\begin{cases}
    - \partial_t u &= L_t [u]\mathbf{V}^{\mathbf{Y}}_t + \Gamma_t[u] \dot{Y}^{\epsilon}_t, \\ u (T, x) & = \phi (x).
    \end{cases}\label{classicalPDE}
\end{equation}
From the Feynman-Kac theorem, we know that any bounded solution $u^{\epsilon}\in C^{1,2}([0,T]\times \mathbb{R})$ to this PDE has the unique representation
\begin{equation}
  u^{\epsilon} (t, x) = E [\phi (X^{t, x,\epsilon}_T)] \label{FeynmanKac},
\end{equation}
where $X^{t, x,\epsilon}$ is the unique strong solution to the SDE \begin{equation*}
 X^{t,x,\epsilon}_t=x, \quad dX_s^{t, x,\epsilon} = g(s, X_s^{t, x,\epsilon}) \mathbf{v}_s^{\mathbf{Y}} dB_s + \left (f_0 (s, X_s^{t, x,\epsilon}) \mathbf{V}^{\mathbf{Y}}_s + f (s, X_s^{t, x,\epsilon})
   \dot{Y}^{\epsilon}_s\right ) ds.
\end{equation*} \begin{remark}
    It is worth to recall that there are situations where we cannot expect solutions to \eqref{classicalPDE} to be in $C^{1,2}$, and one needs a notion for \textit{weak solutions}. However, if $u^{\epsilon}$ defined in \eqref{FeynmanKac} is bounded and continuous on $[0,T]\times \mathbb{R}$, then it still represents a weak-solution to the PDE \eqref{classicalPDE}, and it is in fact the unique \textit{viscosity solution}.We refer to \cite{crandall1992user} for an exposition of viscosity solutions, and details of the stochastic representation of viscosity solutions can be found in \cite{fleming2006controlled}.
\end{remark} For the purpose of this paper, it is enough to know that the Feynman-Kac representation in \eqref{FeynmanKac} is the suitable notion for a general solution to the PDE \eqref{classicalPDE}. In the following theorem we prove that there exists a unique limit of $u^{\epsilon}$ as $\epsilon \to 0$, which we define to be the unique solution to the RPDE \eqref{RPDECl}.

\begin{theorem}
  \label{RPDEmain}Let $\alpha \in (1 / 3, 1 / 2]$ and $\mathbf{Y}\in \mathscr{C}^{\alpha,1+}([0,T],\mathbb{R})$. Consider a sequence $Y^{\epsilon}\in \mathrm{Lip}([0,T],\mathbb{R})$, with (canonical) rough path lift $\mathbf{Y}^{g,\epsilon}$, such that $\tmmathbf{Y}^{g, \epsilon} \rightharpoonup \tmmathbf{Y}^g$ in $\mathscr{C}^{\alpha}$. Moreover, let $f,g\in C^3_b([0,T]\times \mathbb{R})$ and $\phi \in C_b^0(\mathbb{R},\mathbb{R})$. Let $u^{\epsilon}$ be the unique bounded Feynman-Kac (viscosity) solution
  \eqref{FeynmanKac} with respect to $ Y^{\epsilon}$. Then there exists a function $u^{\mathbf{Y}} = u^{\mathbf{Y}} (t, x)
  \in C_b^0 ([0,T]\times \mathbb{R}, \mathbb{R})$, only depending on $\mathbf{Y}$ but not on its approximation, such that $u^{\epsilon} \to u^{\mathbf{Y}}$ pointwise, with Feynman-Kac representation \begin{equation*}
      u^{\mathbf{Y}}(t,x) = E[\phi(X_T^{t,x,\mathbf{Y}})].
  \end{equation*}Moreover, the solution map \begin{align*}
      S: C^0_b (\mathbb{R}, \mathbb{R})
  \times \mathscr{C}^{\alpha,1+}([0,T],\mathbb{R}) & \longrightarrow C^0_b ([0,T]\times \mathbb{R}, \mathbb{R}) \\ (\phi,\mathbf{Y}) & \longmapsto u^{\mathbf{Y}}
  \end{align*}
  is continuous.
\end{theorem}

\begin{proof}
First denote by $X^{t,x,\mathbf{Y}}$ the unique solution to the RDE \eqref{RDEStr}, see Theorem \ref{consistencycoro}. It is discussed in Appendix \ref{RSDEAPPendix} Remark \ref{remarkhybridinter}, that we can equivalently write the RSDE \eqref{RSDErep4} as
  \begin{equation*}
  X_t^{t, x, \tmmathbf{Y}}  =  x, \quad dX^{t, x, \tmmathbf{Y}}_s  =  g(s, X^{t, x, \tmmathbf{Y}}_s)
  \mathbf{v}^{\mathbf{Y}}_s dB_s + f_0 (s, X_s^{t, x, \tmmathbf{Y}}) \mathbf{V}^{\mathbf{Y}}_sds
  + f \left( {s, X_s^{t, x, \tmmathbf{Y}}}  \right)d\mathbf{Y}^g_s
\end{equation*} where $f_0(t,x):=-\frac{1}{2}f(t,x)\partial_xf(t,x)$. Denote by $\mathbf{Y}^{\epsilon}$ the rough path in $\mathscr{C}^{\alpha,1+}$ such that $\mathbf{Y}^{\epsilon} =(Y^{\epsilon},\mathbb{Y}^{g,\epsilon}-\frac{1}{2}\delta[\mathbf{Y}]) = (Y^{\epsilon},\mathbb{Y}^{g,\epsilon}-\frac{1}{2}\delta\int \mathbf{V}_t^{\mathbf{Y}}dt)$, see Lemma \ref{LipRPBijectionLemma}. From the same lemma we know that $\Vert \mathbf{Y}^{\epsilon}-\mathbf{Y}\Vert_{\alpha',1+} \longrightarrow 0$ as $\epsilon \to 0$. From \cite[Chapter 3 Theorem 7]{protter2005stochastic} we know that the following (classical) SDE \begin{align*}
      X_t^{t, x, \epsilon} =  x, \quad dX_s^{t, x, \epsilon} = g(s, X^{t, x,
    \epsilon}_s) \mathbf{v}_s^{\mathbf{Y}}dB_s + \left( f_0 (s,
    X_s^{t, x, \epsilon}) \mathbf{V}^{\mathbf{Y}}_s + f \left( {s, X_s^{t, x,
    \epsilon}}  \right) \dot{Y}_s^{\epsilon} \right) ds ,
  \end{align*} has a unique strong solution. Applying Lemma \ref{MarkovProperty}, we have $X^{t,x,\epsilon} = X^{t,x,\mathbf{Y}^{\epsilon}}$ almost surely, where $X^{t,x,\mathbf{Y}^{\epsilon}}$ is the unique solution to the RDE \eqref{RDEStr}, where we replace $\mathbf{Y}$ by $\mathbf{Y}^{\epsilon}$, and we have
  \begin{equation}
    X ^{t, x, \tmmathbf{Y}^{\epsilon}} \longrightarrow X^{t, x, \tmmathbf{Y}}
    \tmop{as} \epsilon \downarrow 0, \label{ucp}
    \end{equation}
   uniformly on compacts in probability.
  Since $\phi$ is bounded, the random variables $\phi (X_T^{t, x,
  \tmmathbf{Y}^{\epsilon}})$ and $\phi (X_T^{t, x, \tmmathbf{Y}})$ are
  integrable (uniformly in $\epsilon$), and we can define
  \[ u^{\mathbf{Y}} = u^{\mathbf{Y}} (t, x) \assign E [\phi (X_T^{t, x, \tmmathbf{Y}})].\] Combining this with \eqref{ucp}, it
  readily follows that $u^{\epsilon} \longrightarrow u^{\mathbf{Y}}$ pointwise. \\
  Finally, consider the solution map $S(\phi, \mathbf{Y}) =
  u^{\mathbf{Y}}$. Then for any sequence $(\phi^n, \tmmathbf{Y}^{ n})$ converging to
  $(\phi, \tmmathbf{Y})$, we want to show that \begin{equation*}
      u^n = E [\phi^n
  (X_T^{\cdot, \cdot, \tmmathbf{Y}^n}) ] \longrightarrow u^{\mathbf{Y}} = E[\phi 
  (X_T^{\cdot, \cdot, \tmmathbf{Y}})]
  \end{equation*} with respect to the supremum norm on $C^0_b
  ([0,T]\times \mathbb{R}, \mathbb{R})$. But this follows from exactly the same
  arguments as above, namely by constructing the RDE \eqref{RDEStr} for both
  $\mathbf{Y}^n$ and $\mathbf{Y}$, and using the same
  stability arguments.
\end{proof} \begin{remark}
    The pointwise convergence $u^{\epsilon}\rightarrow u^{\mathbf{Y}}$ in the last theorem may be improved to locally uniform convergence. This could be achieved by applying an Arzelà-Ascoli argument, together with local Lipschitz estimates of $(t,x,\mathbf{Y}) \mapsto X^{t,x,\mathbf{Y}}$. The latter, however, needs certain uniform integrability estimates for the joint-lift $\mathbf{Z}^{\mathbf{Y}}$, appearing from the Lipschitz constants in RDE stability results. For the joint-lift with Brownian motion $(B,Y)$, this was for instance done in \cite[Theorem 8]{diehl2015levy}, see also \cite[Theorem 9]{friz2020course}.
\end{remark}
We end this section with a brief discussion about the comparison of RPDEs and SPDEs for the interested reader.
\begin{remark} 
 Formally, replacing $d \mathbf{Y}^g$  in \eqref{RPDECl} with $\circ d I$, the Stratonovich differential of our local martingale $I$, leads to (terminal value) SPDE of the form
 \begin{equation}\label{RPDEintro3}
    - d_t u^{\mathbf{}} =  {L_{t}}  [u^{\mathbf{}}] v_t^2 dt + \Gamma_t [u^{\mathbf{}}] \circ d I, \qquad u^{\mathbf{}} (T, x)  =  \phi (x) .
\end{equation}
Remarkably, the SPDE \eqref{RPDEintro3} is, to the best of our knowledge, beyond existing SPDEs theory -- despite the fact that the Markovian case, say $v_t \equiv 1$ (otherwise absorb $v=v(t,x)$ into the
coefficients of $L_t$) is very well-understood. Indeed, in this case \eqref{RPDEintro3} simplifies to%
\footnote{Strictly speaking, one should write $\circ d \overleftarrow{W}$ in \eqref{RPDEintro4} to emphasize the use of {\em backward} Stratonovich interpretation.} 
\begin{equation}\label{RPDEintro4}
       - d_t u^{\mathbf{}} =  {L_{t}}  [u^{\mathbf{}}] dt + \Gamma_t [u^{\mathbf{}}] \circ d W, \qquad u^{\mathbf{}} (T, x) =  \phi (x).
\end{equation}
Linear SPDEs of similar type (``Zakai equation'') have been studied for decades in non-linear filtering theory. By using {\em backward} (Stratonovich) integration against $W$, one can give honest meaning to \eqref{RPDEintro4}, with {\em backward} adapted solution, meaning that $u=u(t,x; \omega)$ is measurable w.r.t. $\mathcal{F}^W_{t,T} = \sigma ( W_v - W_u : t \le u \le v \le T)$. The issue with \eqref{RPDEintro3} is, complications from {\em forward} adapted $(v_t)$ aside, that there is no general backward integration against local martingales\footnote{... owing to the fact that the class of semimartingales is not stable under time-reversal ...}, rendering previous approaches useless. Thus, to the best of our knowledge, the rough paths approach of \eqref{RPDECl} is the only way to make sense of this SPDE  \eqref{RPDEintro3}. 



%
As is common in mathematical finance, the SDE dynamics of \eqref{SDErep4} is naturally given in It\^o, not in Stratonovich, form. Yet, all SPDEs (resp. RPDEs) have been written with noise in Stratonovich (resp. geometric rough path) type. At least when $v \equiv 1$, the reason is easy to appreciate. If one rewrites \eqref{RPDEintro4} in backward It\^o form,
 $$
          - d_t u^{\mathbf{}} =  {\tilde{L}_{t}}  [u^{\mathbf{}}] dt + \Gamma_t [u^{\mathbf{}}] d \overleftarrow{W}, \qquad u^{\mathbf{}} (T, x) =  \phi (x) ,
 $$ 
 its well-posedness would require for $\tilde{L}$ to satisfy a so-called {\em stochastic parabolicity} condition, see for instance \cite{krylov2007stochastic} for the latter. Having Stratonovich, resp. geometric rough path, noise in \eqref{RPDEintro4}, resp. \eqref{RPDECl}, bypasses such complications.

\end{remark}

\subsection{European pricing with RPDEs}\label{pricingRPDE}
In this section we use the Feynman-Kac representation obtained in Theorem \ref{RPDEmain}, to price European options in local stochastic volatility models. More precisely, for some $(t,x) \in [0,T]\times \mathbb{R}$, let $X^{t,x}$ be the unique strong solution to the SDE \eqref{SDErep4}, and consider a bounded and continuous payoff function $\phi \in C_b(\mathbb{R},\mathbb{R})$. Our goal is to compute the prices \begin{equation*}
    (t,x) \mapsto E[\phi(X_T^{t,x})].
\end{equation*} Since the $(\mathcal{F}_t^W)$-adapted volatility process $v$ is possibly non-Markovian, we consider the conditional LSV-dynamics, conditioned on the Brownian motion $W$, to disentagle $v$ from the dynamics $X$. In Section \ref{conditioningsection} we discussed how to rigorously do so, and we derived an $(\mathcal{F}_t^B)$-Markov property for the conditional dynamics, see Theorem \ref{consistencycoro}. To make use of the Markovian nature in $B$, that is to apply the Feynman-Kac result Theorem \ref{RPDEmain}, we study \textit{conditional prices}
\begin{equation*}
  u (t, x, \omega) = E [\phi (X_T^{t, x}) | \mathcal{F}_T^W \nobracket](\omega)
  \label{randomfield} .
\end{equation*} The following theorem is the main result of this section, and it shows that $u$ is the pathwise solution to the RPDE \eqref{RPDECl}.
\begin{theorem}
  \label{RPDEpricethm}Let $g, f \in C^3_b ([0,T]\times \mathbb{R}, \mathbb{R})$,
  $\phi \in C_b^0(\mathbb{R},\mathbb{R}_+)$ and let Assumption \ref{ass1} hold. Moreover, for any $\mathbf{Y}=(Y,\mathbb{Y})\in \mathcal{C}^{\alpha,1+}([0,T],\mathbb{R})$, we denote by $u^{\mathbf{Y}} = u^{\mathbf{Y}} (t, x)$
  the solution to the RPDE \eqref{RPDECl}, given by Theorem \ref{RPDEmain}. Then we have $u (t, x, \omega) = u^{\mathbf{Y}}
  (t, x) |_{\tmmathbf{Y} = \tmmathbf{I} (\omega)} \nobracket$
  almost surely.
\end{theorem}

\begin{proof}
  From Theorem \ref{consistencycoro} we know that $\mathrm{Law}(X^{t, x,\tmmathbf{I}}|\mathcal{F}^W_T) = \mathrm{Law}( X^{t, x}|\mathcal{F}^W_T)$, where $X^{t, x, \tmmathbf{Y}}$ is the
  unique solution to the RDE \eqref{RDEStr}. Thus we have $u (t, x, \omega) = E [\phi (X_T^{t, x, \tmmathbf{I}}) |
  \mathcal{F}^W_T \nobracket](\omega)$ almost surely, and moreover we know that $X^{t,x,\mathbf{I}}$ is adapted to 
  the filtration generated by the (random) rough path $\mathbf{I}$ and the Brownian motion $B$.
Therefore, there exists a Borel measurable function \begin{equation*}
      F: \mathscr{C}^{\alpha,1+}([0,T],\mathbb{R})\times  C^0([0,T],\mathbb{R}) \longrightarrow \mathbb{R},
  \end{equation*} such that $\phi  (X_T^{t,x,\mathbf{I}}) = F(\mathbf{I},B)$, and we can consider $\mathbf{I}$, resp. $B$, as random variables with values in $\mathscr{C}^{\alpha,1+}([0,T],\mathbb{R})$, resp. $C^0([0,T],\mathbb{R})$. By independence of $B$ and $W$, the law of $B$ on $C^0$, denoted by $\mu_B$, defines a regular conditional distribution\footnote{See \cite[Chapter 8 p.167]{kallenberg1997foundations} for the definition of regular conditional distributions.} for $B$ given $\mathcal{F}^W_T$. Applying the disintegration formula for conditional expectations, see for instance \cite[Theorem 8.5]{kallenberg1997foundations}, we have\begin{align*}
      E \left [F(\mathbf{I},B)| \mathcal{F}_T^W\right](\omega) & = \int_{C^0([0,T],\mathbb{R})}F(\mathbf{I}(\omega),b)\mu_B(db) \\ & = \left. \int_{C^0([0,T],\mathbb{R})}F(\mathbf{Y},b)\mu_B(db)\right |_{\mathbf{Y}= \mathbf{I}(\omega)} \\ & = \left. E[\phi(X_T^{t,x,\mathbf{Y}})]\right |_{\mathbf{Y}= \mathbf{I}(\omega)} = u^{\mathbf{I}(\omega)}(t,x),
  \end{align*} where the last equality follows from Theorem \ref{RPDEmain}.
\end{proof}\\
As a direct corollary of Theorem \ref{RPDEpricethm}, we can now present a pricing formula for European claims, using our pathwise solution to the RPDE \eqref{RPDECl}. Let $t=0$ and $x_0 \in \mathbb{R}$, and denote by $X:=X^{0,x_0}$ the unique solution to the SDE \eqref{SDErep4}, starting from $x_0$ at time $t=0$. \begin{corollary}
  \label{priceformula} Under Assumption \ref{ass1}, for $f,g \in C_b^3([0,T]\times \mathbb{R},\mathbb{R})$ and $\phi \in C^0_b(\mathbb{R},\mathbb{R}_+)$, we have
  \begin{equation*}
    y_0 \assign E [\phi (X_T)] = E [u^{\mathbf{I}} (0, x_0) \nobracket] .
  \end{equation*}

\end{corollary} This suggests that to compute the price of the payoff $\phi(X_T)$ at time $t=0$, we can take the expected value of the evaluations at $(0,x_0)$ of the pathwise solutions $u^{\mathbf{I}}$ to the random RPDE \eqref{RPDECl} with $\mathbf{Y}= \mathbf{I}$. \\ \\ Another natural question is whether we can utilize the pathwise solutions $u^{\mathbf{I}}$ to the RPDE to characterize the price at any time $t \in [0,T]$. Specifically, we aim to compute the conditional expectation $E[\phi(X_T)|\mathcal{F}_t]$, where $\mathcal{F}$ represents the full filtration, defined as $\mathcal{F}_t:=\mathcal{F}_t^{B}\lor \mathcal{F}_t^W$. The following lemma establishes the relationship between this conditional expectation and our RPDE solution $u^{\mathbf{I}}$. \begin{proposition}
    \label{CBLemma}
     Under Assumption \ref{ass1}, for $f,g \in C_b^3([0,T]\times\mathbb{R},\mathbb{R})$ and $\phi \in C^0_b(\mathbb{R},\mathbb{R}_+)$ we have \begin{equation*}
        E[\phi(X_T)|\mathcal{F}_t] =  \left. E[u^{\mathbf{I}}(t,x)|\mathcal{F}_t^{W}] \right |_{x=X_t} \quad \text{a.s. for all }t\in [0,T].\label{measurablelem}
    \end{equation*}
\end{proposition}
\begin{proof} First, we can notice that $X_T := X_T^{0,x_0} = X_T^{t,X_t^{0,x_0}}$ almost surely. Indeed, using once again \cite[Chapter 3 Theorem 7]{protter2005stochastic}, we know that $X^{t,x}$ is the unique solution to the integral equation\begin{equation*}
    \begin{aligned}
    X_s^{t, x}  = x + \int_t^sf (u,
  X_u^{t, x}) v_udW_u + \int_t^sg(u, X_u^{t, x}) v_u dB_u , \quad t \leq s \leq T.
    \end{aligned}
\end{equation*} Therefore, the process $\tilde{X}:= X^{t,X_t}$ is the unique solution to the integral equation \begin{align*}
    \tilde{X}_s  & = X_t + \int_t^sf (u,
  \tilde{X}_u) v_udW_u + \int_t^sg(u, \tilde{X}_u) v_u dB_u, \quad t \leq s \leq T.
\end{align*} But on the other hand, for $s\in [t,T]$, we have \begin{align*}
    X_s & = x_0 + \int_0^sf (u,
  X_u) v_udW_u + \int_0^sg(u, X_u) v_u dB_u \\ & = X_t + \int_t^sf (u,
  X_u) v_udW_u + \int_t^sg(u, X_u) v_u dB_u.
\end{align*} By uniqueness, it follows that $X= \tilde{X}$ on $[t,T]$ and in particular $X_T=X_T^{t,X_t}$. \\ Next, we recall from Theorem \ref{consistencycoro} that \begin{equation*}
\mathrm{Law}\left (\left. X^{t,x} \right |\mathcal{F}^W_T\lor \mathcal{F}^B_t\right) = \mathrm{Law}\left (\left. X^{t,x,\mathbf{I}}\right |\mathcal{F}^W_T\lor \mathcal{F}^B_t\right).
\end{equation*} By continuity of the RDE solution map $x \mapsto X^{t,x,\mathbf{I}}$ we can find a well-defined version of $x \mapsto E[\phi(X_T^{t,x,\mathbf{I}})|\mathcal{F}^W_T \lor \mathcal{F}_t^B]$. Applying the tower property and the Markov property for $X^{t,x,\mathbf{I}}$ yields\begin{equation*}E[\phi(X_T)|\mathcal{F}_t]= E[\phi(X_T^{t,X_t})|\mathcal{F}_t](\omega) = E[E[\phi(X_T^{t,X_t,\mathbf{I}})|\mathcal{F}^W_T\lor \mathcal{F}^B_t]|\mathcal{F}_t]=E[E[\phi(X_T^{t,x,\mathbf{I}})|\mathcal{F}^W_T]|\mathcal{F}^W_t]|_{x=X_t}.
\end{equation*} From Theorem \ref{RPDEpricethm} we know that $E[\phi(X_T^{t,x,\mathbf{I}})|\mathcal{F}_T^W](\omega)=u^{\mathbf{I}(\omega)}(t,x)$, and thus we can conclude \begin{equation*}
    E[\phi(X_T)|\mathcal{F}_t] = \left. E[u^{\mathbf{I}}(t,x)|\mathcal{F}_t^W]\right|_{x=X_t} \quad \text{almost surely.}
\end{equation*}
\end{proof}
\subsection{Special case: stochastic volatility models}\label{RVolsection}
In this section, we discuss two choices of $f$ and $g$, where we can solve the RPDE in Theorem \ref{RPDEpricethm} explicitly. These examples correspond to stochastic volatility models in two distinct forms.
\subsubsection*{Example 1: Stochastic volatility models in Bachelier-form}
First, let $f(t,x) \equiv \rho$ and $g(t,x) \equiv \sqrt{1-\rho^2}$ for some $\rho \in [-1,1]$. In this case we have a stochastic volatility model of the form
\begin{equation*}
  X^{t,x}_t=x,\quad dX^{t,x}_s = v_s \left( \rho dW_s + \sqrt{1 - \rho^2}
dB_s \right), \quad 0 \leq t <s \leq T, \label{priceRB}
\end{equation*}
where $B$ and $W$ are two independent standard Brownian motions, and $v$ satisfies Assmuption \ref{ass1}. The goal now is to show that the pricing RPDE \eqref{RPDECl} reduces to a classical heat equation, and we can give an explicit solution. Let $\phi : \mathbb{R} \rightarrow
\mathbb{R}_+$ be a bounded and continuous function. Since $f$ is constant, we have $f_0 \equiv 0$, and the RPDE \eqref{RPDECl} reduces to \begin{equation}
\begin{cases}
    - d_t u^{\mathbf{Y}} & =  \frac{1}{2}(1-\rho^2)\partial_{xx}^2u^{\mathbf{Y}}_t\mathbf{V}^{\mathbf{Y}}_tdt  + \rho \partial_xu^{\mathbf{Y}}_t d\mathbf{Y}^g_t \\ u^{\mathbf{Y}}(T, x) & =  \phi (x). 
    \end{cases} \label{RPDESVM}
\end{equation}
Choosing the pair $\mathbf{Y}^{g,\epsilon}$, resp. $\mathbf{Y}^{\epsilon}$, as described in the proof of Theorem \ref{RPDEmain}, we know the solution to the RPDE is defined as the unique limit of the sequence of solutions $u^{\epsilon}$ to the PDEs \begin{equation}
\begin{cases}
    - \partial_t u^{\epsilon} & =  \frac{1}{2}(1-\rho^2)\mathbf{V}^{\mathbf{Y}}_t\partial_{xx}^2u^{\epsilon}_t  + \rho  \dot{Y}^{\epsilon}_t \partial_xu^{\epsilon}_t\\u^{\epsilon} (T, x) & =  \phi (x).\label{RPDERT}
    \end{cases} 
\end{equation} But in this case, the PDE can actually be reduced to the heat equation \begin{equation}
\begin{cases}
    - \partial_t v & =  \frac{1}{2}(1-\rho^2)\mathbf{V}^{\mathbf{Y}}_t\partial_{xx}^2v_t \\v (T, x) & =  \phi (x).
    \end{cases} \label{heatequationRPDE}
\end{equation} Indeed, assume $v$ solves \eqref{heatequationRPDE} and define $u^{\epsilon}(t,x) :=v(t,x+\rho Y_{t,T}^{\epsilon})$. Then clearly $u^{\epsilon}(T,x) = v(T,x) = \phi(x)$ and by the chain rule,\begin{align*}
      - \partial_t u^{\epsilon} & =  - \partial_t v(t, y) |_{y = x +
   \rho Y^{\epsilon}_{t,T}} + 
    \rho\dot{Y}^{\epsilon}_t \partial_x v(t,y)|_{y = x +
   \rho Y^{\epsilon}_{t,T}} \\ & = \frac{1}{2} (1-\rho^2)\mathbf{V}^{\mathbf{Y}}_t \partial_{xx}^2 u^{\epsilon}_t
    + \rho \dot{Y}^{\epsilon}_t \partial_x u_t^{\epsilon}.
  \end{align*} Therefore, $u^{\epsilon}$ is indeed the solution to \eqref{RPDERT}. \\ Focusing on the PDE \eqref{heatequationRPDE}, we define the heat-kernel\begin{equation*}
    f(t, x, y) \assign \frac{1}{\sqrt{2 \pi(1-\rho^2) [\mathbf{Y}]_{t,T}}}
   \exp \left\{ - \frac{(y - x)^2}{2 (1-\rho^2) [\mathbf{Y}]_{t,T}} 
   \right\}.
\end{equation*} Then one can check that the function \begin{equation*}
  v(t, x) = \int_{\mathbb{R}} \phi (y) f(t, x, y)^{}dy
  \label{explicitsol}
\end{equation*}
is the unique solution to the heat equation \eqref{heatequationRPDE}. Therefore, the solution to the PDE \eqref{RPDERT} is given by \begin{equation*}
    u^{\epsilon}(t,x) := v(t, x+\rho Y_{t,T}^{\epsilon}) = \int_{\mathbb{R}} \phi (y) \frac{1}{\sqrt{2 \pi(1-\rho^2) [{\mathbf{Y}}]_{t,T}}}
   \exp \left\{ - \frac{(y - x-\rho Y_{t,T}^{\epsilon})^2}{2 (1-\rho^2)[{\mathbf{Y}}]_{t,T}} 
   \right\}dy.
\end{equation*} Using dominated convergence, we finally find the solution to the RPDE is explicitely given by \begin{equation*}
    u^{\mathbf{Y}}(t,x) = \lim_{\epsilon \to 0}u^{\epsilon}(t,x) = \int_{\mathbb{R}} \phi (y) \frac{1}{\sqrt{2 \pi(1-\rho^2) [{\mathbf{Y}}]_{t,T}}}
   \exp \left\{ - \frac{(y - x-\rho Y_{t,T})^2}{2 (1-\rho^2)[{\mathbf{Y}}]_{t,T}} 
   \right\}dy.
\end{equation*} By a direct application of Theorem \ref{RPDEpricethm}, we find the following formulas in the case of stochastic volatility models. \begin{theorem}\label{RTTheorem}
    Under Assumption \ref{ass1}, for $\phi \in C_b^0(\mathbb{R},\mathbb{R})$, the unique solution to the RPDE \eqref{RPDESVM} is given by \begin{equation*}
       u^{\mathbf{Y}}(t,x) =  \int_{\mathbb{R}} \phi (y) \frac{1}{\sqrt{2 \pi(1-\rho^2) [{\mathbf{Y}}]_{t,T}}}
   \exp \left\{ - \frac{(y - x-\rho Y_{t,T})^2}{2 (1-\rho^2)[{\mathbf{Y}}]_{t,T}} 
   \right\}dy.
    \end{equation*} Moreover, recalling that $u(t,x,\omega) = E\left [\phi(X_T^{t,x})|\mathcal{F}_T^W\right ](\omega)$, we have \begin{equation*}
        u(t,x,\omega) =  \int_{\mathbb{R}} \phi (y) \frac{1}{\sqrt{2 \pi(1-\rho^2) [I]_{t,T}(\omega)}}
   \exp \left\{ - \frac{(y - x-\rho I_{t,T}(\omega))^2}{2 (1-\rho^2)[I]_{t,T}(\omega)} 
   \right\}dy.
    \end{equation*}
\end{theorem}
\subsubsection*{Example 2: Stochastic volatility models in Black-Scholes form}
A second example comes from choosing $f(t,x) = x\rho$ and $g(t,x) = x \sqrt{1-\rho^2}$, which leads to the more common form of stochastic volatility models \begin{equation*}
  X^{t,x}_t=x, \quad dX^{t,x}_s = X_s^{t,x}v_s \left( \rho dW_s + \sqrt{1 - \rho^2}
dB_s \right),\quad 0 \leq t <s \leq T. \label{priceSVMBS}
\end{equation*} In this case we have $f_0(t,x) = -\frac{1}{2}\rho^2x$, and thus the RPDE reads \begin{equation}
\begin{cases}
    - d_t u^{\mathbf{Y}} & =  \left (\frac{1}{2}x^2(1-\rho^2)\partial_{xx}^2u^{\mathbf{Y}}_t-\frac{1}{2}\rho^2x\partial_xu^{\mathbf{Y}}_t\right)\mathbf{V}^{\mathbf{Y}}_t dt  + \rho x \partial_xu^{\mathbf{Y}}_t d\mathbf{Y}^g_t\\ u^{\mathbf{Y}}(T, x) & =  \phi (x).
    \end{cases} \label{RPDESVMBS}
\end{equation} Similar as before, the classical PDE with respect to the approximating rough path $\mathbf{Y}^{g,\epsilon}$ as described in the proof of Theorem \ref{RPDEmain}, is given by \begin{equation}
\begin{cases}
    - \partial_t u^{\epsilon} & =  \frac{1}{2}x^2(1-\rho^2)\mathbf{V}^{\mathbf{Y}}_t\partial_{xx}^2u^{\epsilon}_t  + \left(\rho x\dot{Y}^{\epsilon}_t-\frac{1}{2}\rho^2x\mathbf{V}^{\mathbf{Y}}_t \right )\partial_xu^{\epsilon}_t\\u^{\epsilon} (T, x) & =  \phi (x).\label{RPDERTapprox}
    \end{cases} 
\end{equation} We can reduce that PDE to \begin{equation}
\begin{cases}
    - \partial_t v & =  \frac{1}{2}x^2(1-\rho^2)\mathbf{V}^{\mathbf{Y}}_t\partial_{xx}^2v_t \\v (T, x) & =  \phi (x).
    \end{cases} \label{heatequationRPDEBS}
\end{equation} Indeed, define $\psi^{\epsilon}(t):=\rho Y^{\epsilon}_{t,T}-\frac{1}{2}\rho^2[\mathbf{Y}]_{t,T}$ and $u^{\epsilon}(t,x):= v(t,xe^{\psi^{\epsilon}(t)})$. Clearly we have $u^{\epsilon}(T,x)= v(T,x)=\phi(x)$, and by the chain rule \begin{align*}
      - \partial_t u^{\epsilon} & =  - \partial_t v(t, y) |_{y = xe^{\psi^{\epsilon}(t)}}- \partial_x v(t,y)|_{y = xe^{\psi^{\epsilon}(t)}}xe^{\psi^{\epsilon}(t)}\dot{\psi}^{\epsilon}(t)\\ & = \left.\left(\frac{1}{2}y^2 (1-\rho^2) \partial_{xx}^2 v(t,y)\right)\right|_{y = xe^{\psi^{\epsilon}(t)}}\mathbf{V}^{\mathbf{Y}}_t -\left. \left(\partial_xv(t,y)e^{\psi^{\epsilon}(t)}\right)\right|_{y = xe^{\psi^{\epsilon}(t)}} x\dot{\psi}^{\epsilon}(t)\\ & = \frac{1}{2}x^2 (1-\rho^2)\mathbf{V}^{\mathbf{Y}}_t \partial_{xx}^2 u^{\epsilon}_t
    + x\left (\rho \dot{Y}_t^{\epsilon} -\frac{1}{2}\rho^2\mathbf{V}^{\mathbf{Y}}_t \right) \partial_x u_t^{\epsilon},
  \end{align*} and therefore $u^{\epsilon}$ is the solution to \eqref{RPDERTapprox}. \\ Looking directly at the PDE \eqref{heatequationRPDEBS}, we can again directly find a solution, which in this case is given by
      \begin{equation*}
  v(t, x) = \int_{\mathbb{R}_+} \phi (y)\frac{1}{y\sqrt{2 \pi(1-\rho^2) [\mathbf{Y}]_{t,T}}}
   \exp \left\{ - \frac{(\mathrm{ln}(y/x)-\frac{1}{2}(1-\rho^2)[\mathbf{Y}]_{t,T})^2}{2 (1-\rho^2)[\mathbf{Y}]_{t,T}} 
   \right\}dy.
  \label{explicitsolBs}
\end{equation*} Using dominated convergence, and applying Theorem \ref{RPDEmain}, we find in this case \begin{align*}
    u^{\mathbf{Y}}(t,x) & = \lim_{\epsilon \to 0}v(t,xe^{\psi^{\epsilon}(t)}) \\ & = \int_{\mathbb{R}_+} \phi (y)\frac{1}{y\sqrt{2 \pi(1-\rho^2) [\mathbf{Y}]_{t,T}}}
   \exp \left\{ - \frac{(\mathrm{ln}(y/x) -\rho Y_{t,T}+\frac{1}{2}[\mathbf{Y}]_{t,T})^2}{2 (1-\rho^2)[\mathbf{Y}]_{t,T}} 
   \right\}dy.
\end{align*} We summarize the results in the following theorem. \begin{theorem}\label{RTTheoremBS}
    Under Assumption \ref{ass1}, for $\phi \in C_b(\mathbb{R},\mathbb{R}_+)$, the unique solution to the RPDE \eqref{RPDESVMBS} is given by \begin{equation*}
        u^{\mathbf{Y}}(t,x) =  \int_{\mathbb{R}_+} \phi (y)\frac{1}{y\sqrt{2 \pi(1-\rho^2) [\mathbf{Y}]_{t,T}}}
   \exp \left\{ - \frac{(\mathrm{ln}(y/x)-\rho Y_{t,T}+\frac{1}{2}[\mathbf{Y}]_{t,T})^2}{2 (1-\rho^2)[\mathbf{Y}]_{t,T}} 
   \right\}dy.
    \end{equation*} Moreover, recalling that $u(t,x,\omega) = E\left [\phi(X_T^{t,x})|\mathcal{F}^W_T\right ](\omega)$, we have \begin{equation*}
        u(t,x,\omega) =  \int_{\mathbb{R}_+} \phi (y)\frac{1}{y\sqrt{2 \pi(1-\rho^2) [I]_{t,T}(\omega)}}
   \exp \left\{ - \frac{(\mathrm{ln}(y/x) -\rho I_{t,T}(\omega)+\frac{1}{2}[I]_{t,T}(\omega))^2}{2 (1-\rho^2)[I]_{t,T}(\omega)} 
   \right\}dy.
    \end{equation*}
\end{theorem}
\begin{remark}\label{RTremark}
  The formulas for $u$ in Theorems \ref{RTTheorem} and \ref{RTTheoremBS} can also be deduced from the so-called Romano-Touzi formula \cite{romano1997contingent}, and are already well-known. In the case where $f$ and $g$ are constant, one can observe that conditional on $W$, we have \begin{equation*}
      X_T^{t,x}=x +\rho \int_t^Tv_sdW_s+ \sqrt{1-\rho^2}\int_t^Tv_sdB_s \sim \mathcal{N}(x+\rho I_{t,T},(1-\rho^2)[I]_{t,T}),
  \end{equation*} which readily gives the representation for $u$ in Theorem \ref{RTTheorem}. Similarly, in the second example, one can notice that conditional on $W$, the random variable $X_T^{t,x}$ has a log-normal distribution, which leads to the latter formula of $u$. The approach presented in this paper therefore generalizes the understanding of the conditional process $X$ given $W$ in generic local stochastic volatility models, and relates them to pathwise PDEs.
\end{remark}

\subsection{Multivariate dynamics}\label{multivariatesection}
We end this section by describing how our approach generalizes for multidimensional dynamics. For $d_X \in \mathbb{N}$, the $d_X$-dimensional LSV price is given by \begin{equation}
    X^{t,x}_t=x\in \mathbb{R}^{d_X}, \quad dX^{t,x}_s = F(s,X_s^{t,x})v_sdW_s+G(s,X_s^{t,x})v_sdB_s \quad t< s \leq T .\label{multidimdynamics_rev}
\end{equation} where \begin{itemize}
\item $F,G:[0,T]\times \mathbb{R}^{d_X} \rightarrow \mathbb{R}^{d_X\times d}$ are smooth local volatility functions, $d \in \mathbb{N}$,
\item $W$ and $B$ are two independent, $d-$dimensional Brownian motions,
\item $(v_t)_{t\in [0,T]}$ is an $\mathbb{R}^{d \times d}-$valued, $(\mathcal{F}^W_t)-$progressive process with bounded samples paths
\end{itemize}

\begin{remark}
    In fact, we can generalize further by allowing $B$ and $W$ to be of different dimensions, and by letting $v$ be a non-square matrix. The corresponding price dynamics can then be defined by introducing linear maps $\Pi^W$ and $\Pi^B$, such that $I = \int v \Pi^W \, dW$ and $M = \int v \Pi^B \, dB$ are local martingales of the same dimension. However, to keep this section concise and clear, we restrict our discussion to the quadratic case described above.

\end{remark}
The component-wise meaning for the dynamics in \eqref{multidimdynamics_rev} is given by $$X_s^{t,x,i}=x^{i}+\sum_{j,k=1}^{d}\int_t^s \left  (F^{ij}(u,X^{t,x}_u)v^{jk}_u d W_u^{k} +G^{ij}(u,X^{t,x}_u)v^{jk}_u dB_u^{k}\right ), \quad 1 \leq i \leq d_X.$$ 
As before, the pair $(F,G)$ is chosen for correlation, leverage effects and calibration, and is assumed to be sufficiently nice for \eqref{multidimdynamics_rev} to have unique strong solutions, which can be deduced once again by \cite[Chapter 3 Theorem 7]{protter2005stochastic}.

\begin{example} Taking $d_X = d$, and some correlation matrix $\rho \in \mathbb{R}^{d_X \times d_X}$, we define $\bar{\rho} = \sqrt{I_{d_X}-\rho\rho^{T}}$. Then the Bachelier, resp. Black-Scholes stochastic volatility models considered in the last section, generalize by choosing $F(t,x)=\rho$ and $G(t,x) =\bar{\rho}$, resp. $F(t,x)=\mathrm{diag}(x)\rho$ and $G(t,x) = \mathrm{diag}(x)\bar{\rho}$. \end{example} 

We can proceed similar as in the one-dimensional case, by defining the local martingale $I_t=\int_0^tv_sdW_s \in \mathbb{R}^{d}$, and lift it to the $\alpha$-Hölder rough path $\mathbf{I}=(\delta I,\int(\delta I)\otimes dI)$, with non-decreasing Lipschitz rough path brackets (see Definition \ref{defLipschitzbrackets}) $[\mathbf{I}]=\int v_tv_t^{T}dt \in \mathbb{S}^{d}_+$. 
Now for any fixed rough path $\mathbf{Y} \in \mathscr{C}^{\alpha,1+}([0,T],\mathbb{R}^{d})$, we have $\mathbf{v}^{\mathbf{Y}}_t = \sqrt{\partial_t{[\mathbf{Y}]_t}} \in \mathbb{R}^{d\times d},$ and it follows that $\mathbf{v}_t^{\mathbf{I}}(\mathbf{v}_t^{\mathbf{I}})^T = v_tv_t^{T}$ almost surely. In the spirit of Section \ref{conditioningsection}, we consider the RSDE $$X^{t,x,\mathbf{Y}}_t=x, \quad dX^{t,x,\mathbf{Y}}_s = F(s,X_s^{t,x,\mathbf{Y}})d\mathbf{Y}_s+G(s,X_s^{t,x,\mathbf{Y}})\mathbf{v}_s^{\mathbf{Y}}dB_s, \quad t < s \leq T.$$ By the multivariate version of Theorem \ref{consistencycoro}, that is Theorem \ref{consistencycoroA} in Appendix \ref{RSDEAPPendix}, the unique solution the the latter equation is an $(\mathcal{F}^B_t)-$Markov process such that for $\mathbf{Y}=\mathbf{I}$, its conditional law, given $\mathcal{F}^W_T\lor \mathcal{F}_t^B$, coincides with the conditional law of the price dynamics \eqref{multidimdynamics_rev}. Moreover, using similar arguments as in Section \ref{FeynmanKacSection}, for any bounded and continuous function $\phi: \mathbb{R}^{d_X} \longrightarrow \mathbb{R}$, we can find a unique Feynman-Kac type of solution to the RPDE \begin{equation}
\begin{cases}
    - d_t u & =  {L_{t}}  [u]dt + \Gamma_t [u] d\mathbf{Y}_t^g \quad t\in [0,T[ \times \mathbb{R}^{d_X},\\ u (T, x, \cdot) & =  \phi (x), \quad x \in \mathbb{R}^{d_X},
    \end{cases} \label{RPDEmulti}
\end{equation} where
\begin{eqnarray*}
  {L } _t[u] (x) & = & \frac{1}{2}\mathrm{Tr}\left [G(t,x)(\mathbf{v}^{\mathbf{Y}}_t)(\mathbf{v}^{\mathbf{Y}}_t)^{T}G(t,x)^T\nabla_{xx}^2u(t,x)\right]+\langle F_0(t,x),\nabla_xu(t,x)\rangle 
  \\\
  \Gamma_t [u] (x) & = & \langle F(t, x), \nabla_x u (t, x)\rangle,
\end{eqnarray*} where we define $F_0(t,x) \in \mathbb{R}^{d_X}$ by $F^{i}_0(t,x) := -\frac{1}{2} \sum_{j,k,l} \partial_{x_k} F^{ij}(t,x)F^{kl}(t,x)(\mathbf{v}^{\mathbf{Y}}_t(\mathbf{v}^{\mathbf{Y}}_t)^T)^{lj}$ for $i=1,\dots,d_X$. Combining these observations, by the same arguments used in Theorem \ref{RPDEpricethm}, the unique solution to the RPDE above, denoted by $u=u^{\mathbf{Y}}(t,x)$, satisfies
\begin{equation*}
    u^{\mathbf{I}(\omega)}(t,x)= E\left [\phi(X_T^{t,x})|\mathcal{F}_T^{W}\right](\omega)
\end{equation*} almost surely.

\section{Numerical examples}\label{Numeriksection}

We first present some ideas for numerical methods to price European options in general local stochastic volatility models, using the theory developed in the last sections. Recall the general price dynamics  
\begin{equation*}
    \begin{aligned}
    X_t^{t, x} = x, \quad dX_s^{t, x} = f(s,X_s^{t,x})v_sdW_s + g(s,X_s^{t,x})v_sdB_s, \quad t< s \leq T,
    \end{aligned}
\end{equation*} 
for $(t,x) \in [0,T] \times \mathbb{R}$, where $v$ satisfies Assumption \ref{ass1}, and $B$ and $W$ are independent Brownian motions. For some bounded payoff function $\phi \in C_b(\mathbb{R},\mathbb{R})$, we further recall $u(t,x,\omega):= E[\phi(X_T^{t,x})|\mathcal{F}^W_T](\omega)$, and we know from Theorem \ref{RPDEpricethm} that $u(t,x,\omega) = u^{\mathbf{I}(\omega)}(t,x)$ almost surely. Here $\mathbf{I}$ denotes the Itô rough path lift of $I_t=\int_0^tv_sdW_s$, and $u^{\mathbf{Y}}$ is the solution to the RPDE \eqref{RPDECl}. First we introduce two finite-difference schemes to solve the RPDE.
\subsection{Finite-difference schemes for RPDEs}\label{sec:FD_schemes} To construct numerical methods for the RPDE, let us first replace the domain $\mathbb{R}$ by some bounded interval $[a,b]$, for some $a<b$, and consider a (for simplicity) uniform space-grid $\mathcal{S}:= \{a=x_0<x_1<\dots < x_N =b\}$ with $\Delta x:= (b-a)/N$, and similarly a uniform time-grid $\mathcal{T}:= \{0=t_0<t_1<\dots < t_J = T\}$ with $\Delta t:= T/J$. Moreover, when replacing $\mathbb{R}$ with $[a,b]$ in the RPDE, we need to add appropriate boundary conditions. In order to understand the numerical error incurred by the discretization rather than the choice of boundary condition, we choose Dirichlet boundary conditions, i.e., $\psi_a,\psi_b:[0,T] \longrightarrow \mathbb{R}$ and set $u(t,a)=\psi_a(t)$ and $u(t,b) = \psi_b(t)$. The RPDE we want to solve numerically is then of the form 
\begin{subequations}
\begin{align}
    - d_t u^{\mathbf{Y}} &=   {L_{t}}  [u^{\mathbf{Y}}] \mathbf{V}^{\mathbf{Y}}_tdt + \Gamma_t [u^{\mathbf{Y}}] d\mathbf{Y}_t^g, \quad (t,x) \in [0,T[ \times ]a,b[\label{eq1}\\ 
    u^{\mathbf{Y}}(t,a) &=\psi_a(t), \quad t \in [0,T[ \\ u^{\mathbf{Y}}(t,b) &=  \psi_b(t),\quad t\in [0,T[ \\
    u^{\mathbf{Y}} (T, x) &= \phi (x), \quad x\in [a,b]\label{eq4}.
\end{align}
\end{subequations}

\subsubsection*{First-order finite-difference scheme}
Recall from Section \ref{FeynmanKacSection} Theorem \ref{RPDEmain}, that the unique solution $u^{\mathbf{Y}}$ to the RPDE \eqref{RPDECl} is the unique limit of the solution $u^{\epsilon}$ to the classical PDE \eqref{classicalPDE}. Our first approach is to numerically approximate $u^{\epsilon}$, leaving a detailed error-analysis for $|u^{\epsilon}-u^{\mathbf{Y}}|$ aside here. Then we proceed by using a finite-difference scheme to solve the classical PDE \eqref{classicalPDE}. In particular, we set $u_j^n:=u^{\epsilon}(t_j,x_n)$, and recall the notation $\frac{d[\mathbf{Y}]_t}{dt}:=\mathbf{V}^{\mathbf{Y}}_t$. Integrating over the time interval $[t_j,t_{j+1}]$ leads to\begin{eqnarray}
    u_j^n = u_{j+1}^n+\int_{t_j}^{t_{j+1}}L_s[u^{\epsilon}](x_n)d[\mathbf{Y}]_s + \int_{t_j}^{t_{j+1}}\Gamma_s[u^{\epsilon}](x_n)d{Y}^{\epsilon}_s.\label{numerikrpdefirtsz}
\end{eqnarray} We replace the space derivatives in the differential operators $L$ and $\Gamma$ by finite difference quotients
\begin{equation*}
    \begin{aligned}
        & L_{t_j}[u^{\epsilon}](x_n) \approx L_j^n:= \frac{1}{2}g^2(t_j,x_n)\frac{u_j^{n+1}+u_j^{n-1}-2u_j^{n}}{(\Delta x)^2}+ f_0(t_j,x_n)\frac{u^{n+1}_j-u^n_j}{\Delta x} \\ & \Gamma_{t_j}[u^{\epsilon}](x_n) \approx \Gamma_j^n := f(t_j,x_n)\frac{u^{n+1}_j-u^n_j}{\Delta x}.
    \end{aligned}
\end{equation*} Note that the usual stability analysis for finite difference schemes would lead to \emph{path-dependent} conditions for $\Delta t$ and $\Delta x$ in the fully explicit case.  Such considerations have been carried out in \cite{seeger2020approximation} for a similar scheme in the context of numerically solving SPDEs. Indeed, the corresponding explicit scheme was observed to be unstable unless very small time-steps were chosen. To improve stability, we instead use a implicit-explicit (IMEX) scheme. More precisely,
we apply a left-point approximation for the middle integral, and a right-point approximation for the last integral in \eqref{numerikrpdefirtsz}, which leads to the mixed implicit-explicit, backward finite-difference scheme\begin{equation*}
     u_j^n = u_{j+1}^{n}  +L_{j}^n[\mathbf{Y}]_{t_j,t_{j+1}} + \Gamma_{j+1}^nY^{\epsilon}_{t_j,t_{j+1}},
\end{equation*} for $0 \leq j \leq J-1, 1 \leq n \leq N-1$, with  boundary conditions \begin{equation}
    u^0_j = \psi_a(t_j), \quad u^N_j= \psi_b(t_j), \quad u^n_J = \phi(x_n).\label{FDboundary}
\end{equation}

\subsubsection*{Second-order finite-difference scheme}

The first order approach appears to be very natural considering our definition of solutions to the RPDE. Indeed, we gave the RPDE only a formal meaning by defining solutions as limits of classical PDE solutions. A more direct interpretation of RPDEs in a slightly different setting is given in \cite{diehl2017stochastic}. Extending the notion of \textit{regular solutions} in \cite[Definition 7]{diehl2017stochastic} to our setting, we say a function $u=u^{\mathbf{Y}}(t,x) \in C^{0,2}$ is a regular solution to the RPDE \eqref{RPDECl}, if \begin{equation}
    u^{\mathbf{Y}}(t,x) = \phi(x) +\int_t^TL_s[u^{\mathbf{Y}}]\mathbf{V}_s^{\mathbf{Y}}ds+\int_t^T\Gamma_s[u^{\mathbf{Y}}]d\mathbf{Y}_s^g\label{solRPDE}
\end{equation} and $(\Gamma[u^{\mathbf{Y}}],\Gamma'[u^{\mathbf{Y}}]):=(\Gamma[u^{\mathbf{Y}}], -\Gamma \left[\Gamma [u^{\mathbf{Y}}]\right]) \in \mathscr{D}_{\mathbf{Y}^g}^{2\alpha}([0,T],\mathbb{R})$, where $\mathscr{D}^{2\alpha}_{\mathbf{Y}^g}$ is the space of \textit{controlled rough paths}. The latter is now a well-defined, rough integral given by\begin{equation*}
    \int_t^T\Gamma_s[u^{\mathbf{Y}}]d\mathbf{Y}^g_s= \lim_{|\pi| \to 0}\sum_{[v,w]\subset \pi}\left(\Gamma_{v}[u^{\mathbf{Y}}]Y_{v,w}- \Gamma_v \left[\Gamma[u^{\mathbf{Y}}]\right]\mathbb{Y}^g_{v,w}\right),
\end{equation*}where                     $\pi$ denotes a subdivision of $[t,T]$, and $|\pi|$ denotes the mesh-size. \begin{remark} While we firmly believe that our RPDEs can be treated in a similar manner as done in \cite{diehl2017stochastic}, 
via forward-backward duality, in combination with \cite{friz2021rough},
addressing this goes beyond the scope of this paper. Therefore, we intentionally keep this task, as well as a full numerical analysis for both finite-difference schemes, open for future research.
\end{remark} 

Recalling that $\mathbb{Y}^g_{s,t}= \frac{1}{2}(Y_{s,t})^2$ in the one-dimensional setting, the rough integral can be formulated in a backward representation, see \cite[Chapter 5.4]{friz2020course}, given by 
\begin{equation*}
    \int_t^T\Gamma_s[u^{\mathbf{Y}}]d\mathbf{Y}_s^g= \lim_{|\pi| \to 0}\sum_{[v,w]\subset \pi}\left(\Gamma_{w}[u^{\mathbf{Y}}]Y_{v,w}+ \frac{1}{2}\Gamma_w \left[\Gamma[u^{\mathbf{Y}}]\right]Y^2_{v,w}\right)\label{backward}.
\end{equation*}
By definition of $\Gamma$, we have \begin{align*}
    \Gamma_t \left[\Gamma [u^{\mathbf{Y}}]\right](x)& = f(t,x)\partial_xf(t,x)\partial_xu^{\mathbf{Y}}(t,x) + f^2(t,x)\partial_{xx}u^{\mathbf{Y}}(t,x) \\ &= -2f_0(t,x)\partial_xu^{\mathbf{Y}}(t,x) + f^2(t,x)\partial_{xx}u^{\mathbf{Y}}(t,x).
\end{align*} Now with exactly the same idea as before, by integrating over $[t_j,t_{j+1}]$ we have 
\begin{equation*}
    u^{\mathbf{Y}}(t_j,x_n) = u^{\mathbf{Y}}(t_{j+1},x_n)+\int_{t_j}^{t_{j+1}}L_s[u^{\mathbf{Y}}](x_n)\mathbf{V}_s^{\mathbf{Y}}ds+\int_{t_j}^{t_{j+1}}\Gamma_s[u^{\mathbf{Y}}](x_n)d\mathbf{Y}^g_s.
\end{equation*} 
Denoting $u_j^n:=u^{\mathbf{Y}}(t_j,x_n)$ for $0\leq j \leq J$ and $0 \leq n \leq N$, we use the same approximations for $L$ and $\Gamma$ given in the last section, and define \begin{equation*}
    \Gamma'_{t_j}[u](x_n) \approx (\Gamma')_j^n := -2f_0(t_j,x_n)\frac{u^{n+1}_j-u^n_j}{\Delta x}+f^2(t_j,x_n)\frac{u_j^{n+1}+u_j^{n-1}-2u_j^{n}}{(\Delta x)^2}.
\end{equation*} Similar as in the first-order scheme, we apply a left-point approximation for the middle integral, and a right-point approximation for the rough integral in \eqref{solRPDE}, which leads to the mixed implicit-explicit, second order finite-difference scheme \begin{equation*}
    u_j^n = u_{j+1}^{n}  +L_{j}^n[\mathbf{Y}]_{t_j,t_{j+1}} + \Gamma_{j+1}^nY_{t_j,t_{j+1}}+\frac{1}{2}(\Gamma')_{j+1}^nY_{t_j,t_{j+1}}^2,
\end{equation*} for $0 \leq j \leq J-1, 1 \leq n \leq N-1$, with the same boundary conditions as in \eqref{FDboundary}.

\subsection{European options, Greeks and variance reduction}\label{sec:options}

In this section we test our finite-difference schemes to compute option prices and Greeks in two LSV models, and compare the results with full Monte-Carlo simulation. For this we focus on the following two examples \begin{example}[Bachelier stochastic volatility model]\label{LSVexnum}
    Recall from Section \ref{RVolsection}, that for the choice $f\equiv \rho$ and $g\equiv \sqrt{1-\rho^2}$, we find the stochastic volatility model of the form\begin{equation}
    \begin{aligned}
        X^{t,x}_t=x, \quad dX^{t,x}_s & = v_s \left( \rho dW_s + \sqrt{1 - \rho^2}
dB_s \right), \quad t < s \leq T.\label{RTintro}
    \end{aligned}
 \end{equation} As already mentioned in Remark \ref{RTremark}, conditional on $W$, the price has a normal distribution. Applying similar techniques as in the derivation of the Black-Scholes formula, one can find an explicit expression of $u(t,x,\omega)$.
\end{example}

\begin{example}[SABR local stochastic volatility model]\label{LSVex1num}
    In a second example we consider a model leaving the classical stochastic volatility framework, namely we have the SABR dynamics $f(t,x)= \rho x^{\beta}$ and $g(t,x) = \sqrt{1-\rho^2}x^{\beta}$ for some $\beta \in (1/2,1]$, see \cite{hagan2015probability}. More precisely, we have \begin{equation*}
    X^{t,x}_t=x,\quad dX^{t,x}_s = (X^{t,x}_s)^{\beta}v_s \left( \rho dW_s + \sqrt{1 - \rho^2}
dB_s \right), \quad t < s \leq T.
\end{equation*} Notice that this in particular includes the case of classical stochastic volatility models in Black-Scholes form for $\beta =1$, as described in Section \ref{RVolsection}. For $\beta \neq 1$, we no longer have an exact reference solution for $u(t,x,\omega)$, and we use a full Monte-Carlo simulation for comparison.
\end{example}
We are interested in a European put option written on the asset $X$, that is we consider the payoff function $\phi(x):= \max(K-x,0)$. (In case that our model allows for negative price paths, we might need to truncate the payoff for negative $x$ in a smooth way, in order to retain boundedness.) For all the examples in this section, we choose the rough volatility process as given in the rough Bergomi model \cite{bayer2016pricing}, that is \begin{equation*}
    v_t=\xi_0\mathcal{E}\left (\eta \int_0^t(t-s)^{H-\frac{1}{2}}dW_s \right ),
\end{equation*} where $\mathcal{E}$ denotes the stochastic exponential. We choose $T=1, K=5$, and following \cite{bayer2016pricing}, we choose the parameters $\xi_0 = 0.235^2, \eta = 1.9$ and $H=0.07$, where we recall that the Hurst parameter $H$ determines the roughness of the volatility, that is the Hölder regularity of $v$. Then we simulate $M=10\,000$ paths of the pair $(I,[I])=(\int v_tdW_t,\int V_tdt)$, with each $N=10\,000$ uniform time-steps in $[0,1]$, we refer to \cite[Section 4]{bayer2016pricing} for details about simulation in the rough Bergomi model. \\ \\ Now having the samples $(I^{(m)},[I]^{(m)})$ in hand, we can in particular construct samples of the Itô rough path lift $\mathbf{I}^{(m)}$ for $m=1,\dots,M$. Along every such sample path $\mathbf{I}^{(m)}$, we solve the RPDE \eqref{RPDECl} with our finite-difference schemes, and denote by $u^{\mathbf{I}^{(m)}}$ the exact solution along the $m$-th sample. Recalling the identity $u^{\mathbf{I}^{(m)}}(t,x) = E[\phi(X_T^{t,x,\mathbf{I}^{(m)}})]$, see Theorem \ref{RPDEmain}, it is tempting to think of $u^{\mathbf{I}^{(m)}}$ as the put option price written on the $(\mathcal{F}^B_t)$-price dynamics $X^{t,x,\mathbf{I}^{(m)}}$. It is however important to note that the latter process is typically not an $(\mathcal{F}_t^B)$-martingale, and thus $u^{\mathbf{I}^{(m)}}$ could lead to arbitrage opportunities in this model. On the other hand, we can consider the Monte-Carlo approximation\begin{equation*}
    y_0(x) = E[u^{\mathbf{I}}(0,x)] \approx \frac{1}{M}\sum_{m=1}^Mu^{\mathbf{I}^{(m)}}(0,x)=: y_0^M(x).\label{MonteCarloAnsatz}
\end{equation*} The right-hand side then clearly converges to the (arbitrage-free) price $y_0$ in the corresponding LSV model, as $M$ goes to infinity.  \\ \\ For both Examples \ref{LSVexnum} and \ref{LSVex1num}, we denote by $u^{\mathbf{I}^{(m)}}_{\text{FD}}$ the finite-difference solutions to the RPDEs along the samples $m=1,\dots,M$. 


\begin{figure}%
  \centering
  
  {\includegraphics[width=0.98\linewidth]{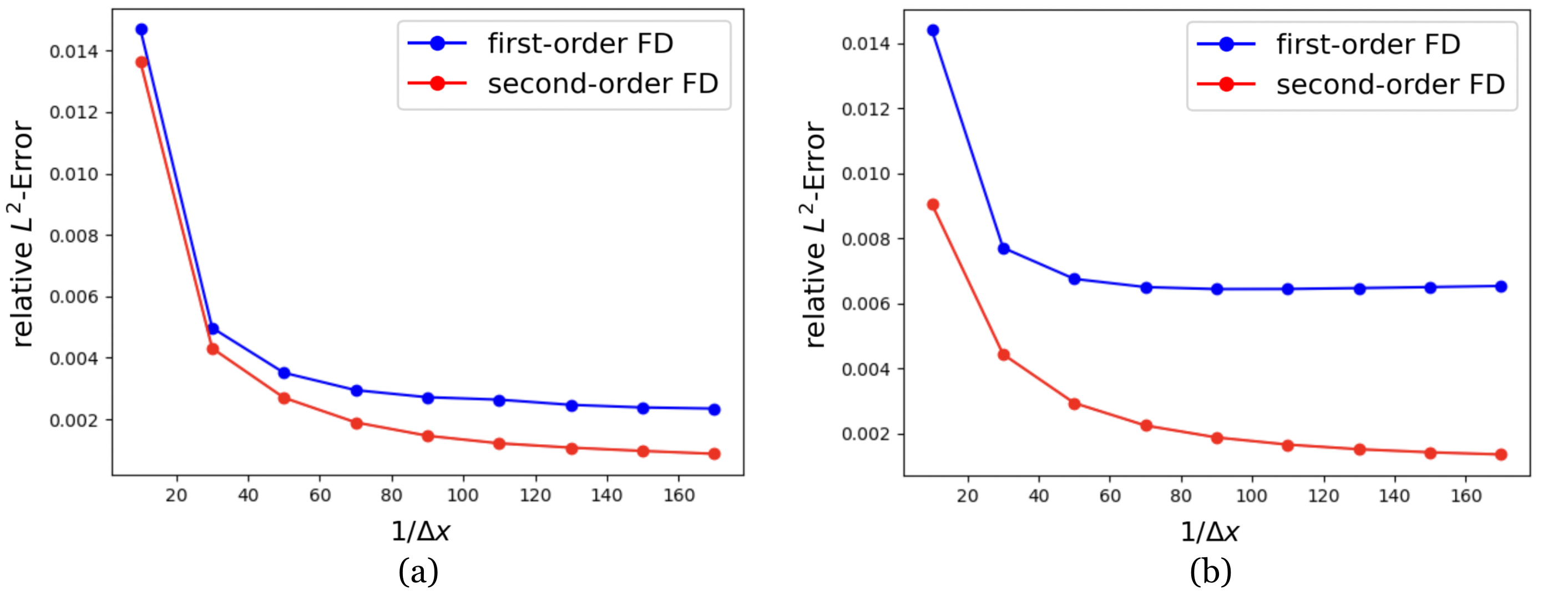}}
  \caption{Strong relative errors for fixed time step-size $\Delta t = 1/30$ and increasing number of space-steps, for both finite-difference scheme. In (a): Bachelier SV-model Example \ref{LSVexnum}, and in (b): SABR SV-model Example \ref{LSVex1num} with $\beta = 0.6$.}%
\label{ok}\end{figure}In Figure \ref{ok} we plot the following strong relative error with respect to the $2$-norm $\Vert h \Vert_{2}:= \sqrt{\sum_{i=0}^Nh(x_i)^2}$ on the space-grid, that is \begin{equation*}
    \epsilon^M:=\frac{1}{M}\sum_{m=1}^M\frac{\Vert u^{\mathbf{I}^{(m)}}(0,\cdot)-u^{\mathbf{I}^{(m)}}_{\text{FD}}(0,\cdot)\Vert _{2}}{\Vert u^{\mathbf{I}^{(m)}}(0,\cdot)\Vert_{2}},
\end{equation*} for both examples and finite-difference schemes. The reference solution in Example \ref{LSVexnum} corresponds to the exact Romano-Touzi solution, and in Example \ref{LSVex1num} we consider a Monte-Carlo solution, obtained along $10\,000$ samples of $(W,B)$. 
    Moreover, we use the reference solutions to specify the boundary conditions in \eqref{eq1}-\eqref{eq4}. It should not come as a surprise that the first-order scheme does not converge to the exact solution $u$, as we recall that this scheme approximates $u^{\epsilon}$ for some fixed $\epsilon$, that is the solution to the classical PDE \eqref{classicalPDE}. Therefore, even if the finite-difference scheme were exact, we would still encounter a fixed error $|u-u^{\epsilon}|$, which is related to accuracy of the piecewise linear approximation $I^{\epsilon}$. In Figure \ref{Samples}, along several samples of $(I,[I])$, we consider the Monte-Carlo solution $u(0,x,\omega)$, obtained via $M=10\,000$ samples with each $N=10\,000$ time-steps, together with both finite-difference solutions along each samples. Moreover, we plot the mean $y_0^M$, which represents an approximation for the price $y_0$. 
    
    \begin{figure}[h]
    \centering
    \includegraphics[width=0.73\textwidth]{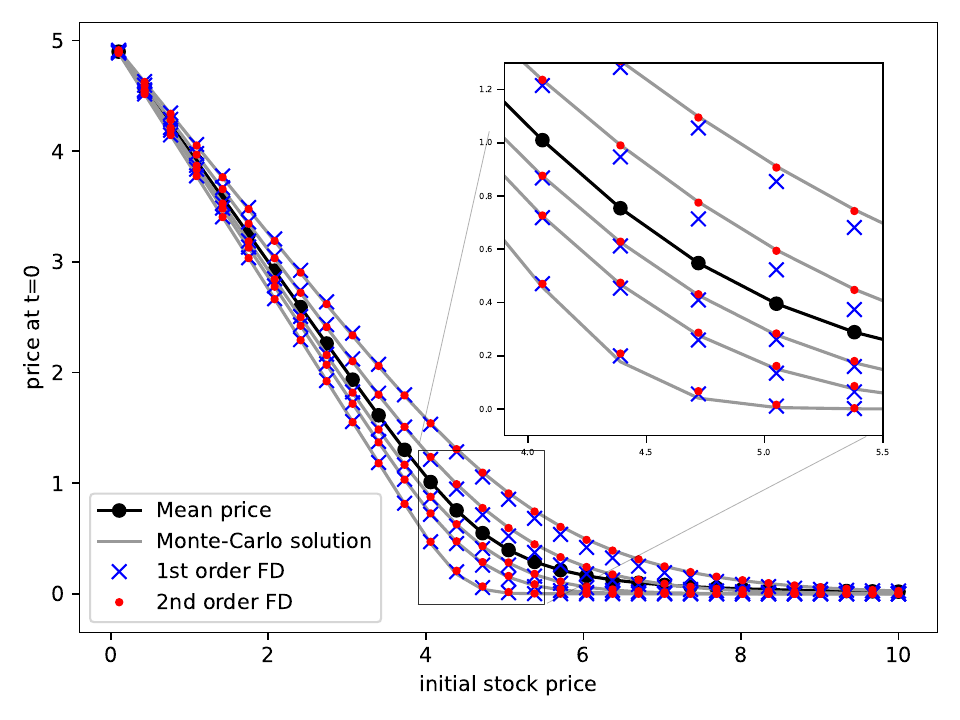}
    \caption{European put option in SABR local stochastic volatility model: Monte-Carlo values of $u^{\mathbf{I}^{(m)}}(0,x)$ vs. finite-difference solutions to the RPDE with first-order (\color{blue}$\times$\color{black}) and second-order (\color{red} $\boldsymbol{\cdot}$\color{black}) schemes, along several sample paths of $(I,[I])$. The black line corresponds to the mean along all $M=10'000$ samples. We choose $\Delta t = 1/120$, $\Delta x = (b-a)/90$, $\rho = -0.4$ and $\beta = 0.6$.}
    \label{Samples}
\end{figure} 
    
Finally, to demonstrate the effects of payoff-regularization and variance reduction in our partial Monte-Carlo approach, we additionally compute the \emph{Delta} and \emph{Gamma} for European put options in the Example \ref{LSVex1num} for the same choices of parameters. More precisely, we consider the following sensitivities with respect to the price \begin{equation}\label{eq:Greeks}
\Delta(x)= \partial_x y_0(x), \quad \Gamma(x)=\partial^2_{xx} y_0(x).
\end{equation} 
Similar as before, we approximate the derivatives by finite-difference quotients $\Delta^h(x)= \frac{1}{h}(y_0(x+h)-y_0(x))$, resp. $\Gamma^h(x)=\frac{1}{h^2}(y_0(x+h)-2y_0(x)+y_0(x-h))$ -- obviously, using the same samples for the computation of all option prices involved. For fixed $h=0.05$, we plot these quantities in Figure \ref{Greeks} for both a full Monte-Carlo simulation of the finite-differences of the payoff, and a partial Monte-Carlo simulation of the finite-differences of the RPDE solutions. For the \emph{Delta}, the first plot in Figure \ref{Greeks} shows similar performance for both methods for the choice of parameters and fixed $h=0.05$. On the other hand, for the $\emph{Gamma}$ we can see that the plain Monte-Carlo method becomes unstable, and smaller $h$ and/or more Monte-Carlo samples are required to get accurate results, while the RPDE approach still appears to be stable. This observation can be explained by the fact that the conditioning in our procedure regularizes the payoff, and by construction reduces the variance. To quantify the latter, we denote by $\sigma^{MC}_{P}(x)$, resp. $\sigma^{MC}_{\Delta^h}(x)$ and $\sigma^{MC}_{\Gamma^h}(x)$, the standard deviations of the put-option, resp. the \emph{Delta} and \emph{Gamma} in the full Monte-Carlo simulation, and similarly we denote by $\sigma^{RPDE}_{P}(x),\sigma^{RPDE}_{\Delta^h}(x)$ and $\sigma_{\Gamma^h}(x)$ the same quantities for the partial Monte-Carlo approach. 
Using the max-norm $\Vert h \Vert_{\infty} = \max_{0\leq j \leq N}|h(x_i)|$, in the performed numerical experiments for Example \ref{LSVex1num} we have\begin{equation}\label{eq:var_reduction}
\left \Vert\frac{\sigma^{MC}_{P}}{\sigma^{RPDE}_{P}}\right \Vert_{\infty} \approx 2.04, \quad \left \Vert\frac{\sigma^{MC}_{\Delta^h}}{\sigma^{RPDE}_{\Delta^{h}}}\right \Vert_{\infty} \approx 5.79, \quad \left \Vert\frac{\sigma^{MC}_{\Gamma^{h}}}{\sigma^{RPDE}_{\Gamma^{h}}}\right \Vert_{\infty} \approx 85.52.
\end{equation} In \eqref{eq:var_reduction} we used the sample standard deviation along the $M=10\,000$ samples, and we observe the claimed reduction, which in particular for the \emph{Gamma} becomes significant. Since the standard deviation represents the constant in the Monte Carlo error, these factors indicate how many more samples are required to achieve the same accuracy when using a full Monte Carlo simulation.
\begin{figure}[h]
    \centering
    \includegraphics[width=1\textwidth]{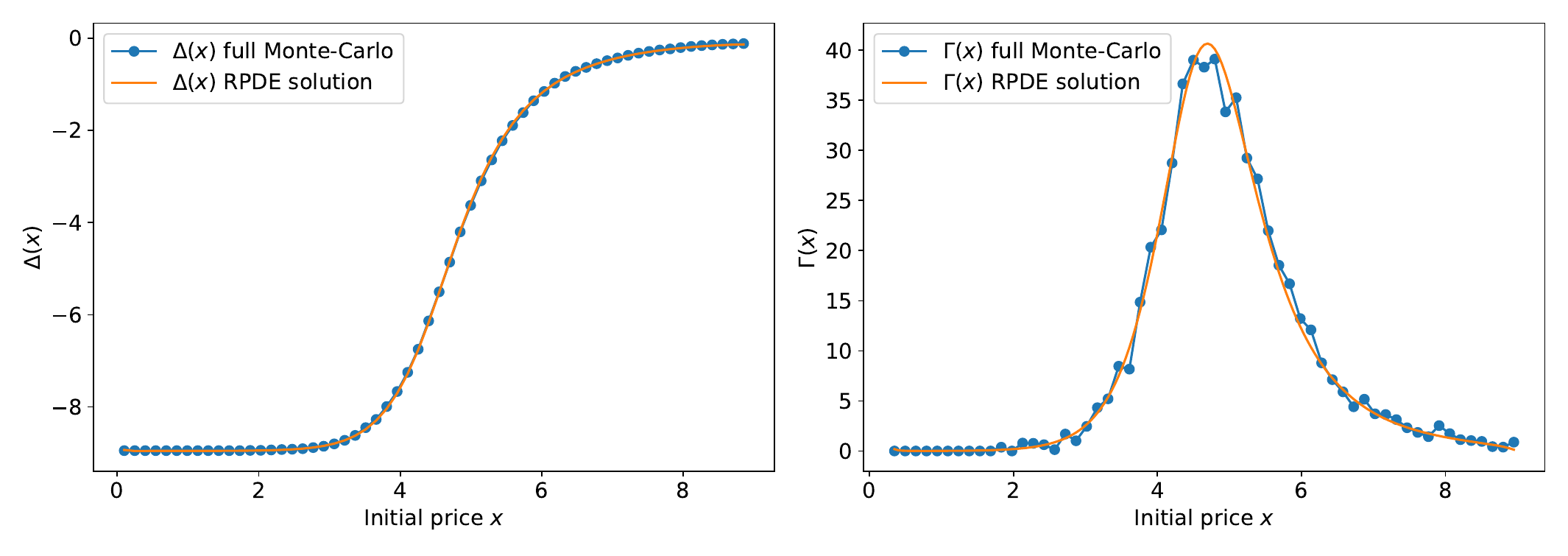}
    \caption{The Greeks \emph{Delta} (left) and \emph{Gamma} (right) for European put options in the SABR model, using full Monte-Carlo simulation vs. partial Monte-Carlo via RPDE solutions. We choose $\Delta t = 1/120$, $\Delta x = (b-a)/90$, $\rho = -0.4$, $\beta = 0.6$ and $h=0.05$.}
    \label{Greeks}
\end{figure}

\subsection*{Data availability statement}
This study is mainly theoretical and does not rely on any real-world datasets. The numerical examples presented in this paper were conducted using synthetic data generated through methods described in the manuscript. The code is available upon request.

\subsection*{Conflict of interest statement}
The authors declare no conflicts of interest.

\appendix
\input{Appendix_final}
\bibliographystyle{plain}
\bibliography{BIB}

\end{document}

%% file: Appendix_final.tex
\section{Appendix}\label{RPappendix} In this section we discuss some technical details from rough path theory that we used in the main results of this paper. For any path $Y:[0,T] \rightarrow V$, recall the increment notations \begin{equation*}
    (\delta Y)_{s,t} = Y_{s,t} := Y_t-Y_s, \quad s \leq t.
\end{equation*}
\subsection{Proofs of Section \ref{Presection}}\label{enhancedAppendix}
 \begin{proof}{\textbf{of Lemma \ref{bijectionlemma} }} Consider the map \begin{equation*}
    \mathcal{T}:\mathscr{C}^{\alpha}_g
  ([0, T], V) \oplus C_0^{2 \alpha} ([0, T], \tmop{Sym} (V \otimes V)) \longrightarrow \mathscr{C}^{\alpha} ([0, T], V),
    \end{equation*} such that $\mathcal{T}((Y,\mathbb{Y}^g),\mathbb{H}):= (Y,\mathbb{Y}^g-\frac{1}{2}\delta \mathbb{H})$. It is straightforward to check that \begin{equation*}
        \Vert \mathbb{Y}^g-\frac{1}{2}\delta \mathbb{H} \Vert_{2\alpha}= \sup_{0\leq s < t \leq T}\frac{|\mathbb{Y}_{s,t}^g-\frac{1}{2}\delta (\mathbb{H})_{s,t}|}{|t-s|^{2\alpha}}< \infty.
    \end{equation*} Moreover, for any $s\leq u \leq t$, we have \begin{align*}
        \mathbb{Y}_{s,t} := \mathbb{Y}^g_{s,t}-\frac{1}{2}(\delta\mathbb{H})_{s,t} &= \mathbb{Y}^g_{s,u}+\mathbb{Y}^g_{u,t} + Y_{s,u} \otimes Y_{u,t} -\frac{1}{2} \delta \mathbb{H}_{s,u}-\frac{1}{2}\delta \mathbb{H}_{u,t} \\ & = \mathbb{Y}_{s,u}+\mathbb{Y}_{u,t}+Y_{s,u} \otimes Y_{u,t},
    \end{align*} where we used the additivity of $\delta \mathbb{H}$ and Chen's relation \eqref{CHEN}. Thus it follows that $\mathcal{T}((Y,\mathbb{Y}^g),\mathbb{H}) \in \mathscr{C}^{\alpha}$. Now to see that $\mathcal{T}$ is injective, consider two elements $\left((Y,\mathbb{Y}^g),\mathbb{H}\right),\left((Y',\mathbb{Y'}^g),\mathbb{H'}\right)$ and assume \begin{equation*}
        \mathcal{T}\left((Y,\mathbb{Y}^g),\mathbb{H}\right)=\mathcal{T} \left((Y',\mathbb{Y'}^g),\mathbb{H'}\right).
    \end{equation*}By definition it follows that $Y=Y'$, and using the geometricity \eqref{geometric}, it follows that two rough path lifts of $Y$, $\mathbb{Y}^g$ and $\mathbb{Y}'^g$, can only differ by some $\delta F: \Delta_{[0,T]}\rightarrow V$, where $\delta F$ has trivial symmetric part. Therefore, write $\mathbb{Y}'^{g} = \mathbb{Y}^g+ \delta F$, we find \begin{align*}
    -\frac{1}{2}\delta \mathbb{H} = -\frac{1}{2}\delta( \mathbb{H'}-2F)
    \end{align*} taking the antisymmetric part, it follows that $F$ is constant, hence $\mathbb{Y}^g= \mathbb{Y}'^g$. Similarly, from $\delta \mathbb{H} = \delta \mathbb{H}'$, it follows that $\mathbb{H}= \mathbb{H}'$. \\ On the other hand, let $\mathbf{Y}\in \mathscr{C}^{\alpha}$, we want to show that there exists a pair $(\mathbf{Y}^g,\mathbb{H}) \in \mathscr{C}^{\alpha}_g
  ([0, T], V) \oplus C_0^{2 \alpha} ([0, T], \tmop{Sym} (V \otimes V))$, such that $\mathcal{T}(\mathbf{Y}^g,\mathbb{H}) = \mathbf{Y}$. But we know that \begin{equation*}
      [\mathbf{Y}]_t:= Y_{0,t}\otimes Y_{0,t}-2\mathrm{Sym}(\mathbb{Y}_{0,t}) \in C_0^{2\alpha}([0,T],\mathrm{Sym}(V \otimes V)).
  \end{equation*} Define $\mathbb{Y}^g:= \mathbb{Y}+\frac{1}{2}\delta [\mathbf{Y}]$, by the same argument as above $(Y,\mathbb{Y}^g)$ defines a rough path, and since \begin{equation*}
      \mathrm{Sym}(\mathbb{Y}^g_{s,t})= \mathrm{Sym}(\mathbb{Y}_{s,t})+\frac{1}{2}Y_{s,t} \otimes Y_{s,t}-\mathrm{Sym}(\mathbb{Y}_{s,t}) = \frac{1}{2}(Y_{s,t}\otimes Y_{s,t}),
  \end{equation*} it is weakly geometric, and therefore the claim follows.
\end{proof} \\ \begin{proof}{\textbf{of Proposition \ref{Mglift}}}
  From the additivity of the It\^{o}-integral, it is straightforward to check \textit{Chen's relation}, see \eqref{CHEN}, for both $\tmmathbf{M}^{\text{It\^{o}}}$ and $\tmmathbf{M}^{\text{Strat}}$. Moreover, applying It\^{o}'s formula to $M \otimes M$, one can notice that $\mathrm{Sym}\left(\mathbb{M}^{\text{Strat}}_{s,t}\right ) = \frac{1}{2}M_{s,t}\otimes M_{s,t}$. \\ Next we consider a sequence of stopping times $(\tau_n)_{n\in \mathbb{N}}$, given by \begin{equation*}
      \tau_n:=\inf\left\{t\geq 0: \Vert \sigma \Vert_{\infty;[0,t]} \geq n\right \}.
  \end{equation*} Since $M$ is continuous, and $\sigma \in L^{\infty}$ almost surely, we have $\tau_n \rightarrow \infty$ as $n\to \infty$. Define the localization $M^n_{\cdot}:= M_{\tau_n \land \cdot}$ and similarly $\mathbb{M}^{\text{It\^{o}},n}$ and $\mathbb{M}^{\text{Strat},n}$. Applying the Burkholder-Davis-Gundy (BDG) inequalities, see for instance \cite[Theorem 18.7]{kallenberg1997foundations}, for all $q\geq 2$ we have \begin{equation}
      E\left [|M^n_{s,t}|^q\right] \lesssim  E\left [\left(\int_{s\land \tau_n}^{t\land \tau_n}|\sigma_u|^2du\right)^{q/2}\right] \lesssim n^q|t-s|^{q/2}.\label{Hoelder1}
  \end{equation} From a classical Kolmogorov continuity criterion, see \cite[Theorem 4.23]{kallenberg1997foundations}, it readily follows that $M^n$ has a Hölder continuous modification for every exponent $\alpha < 1/2$, which for all $q\geq 2$ satisfies \begin{equation}
      E\left [\left(\sup_{0\leq s \leq t \leq T}\frac{|M^n_t-M^n_s|}{|t-s|^{\alpha}}\right)^q\right] < \infty.\label{hoelderbounds}
  \end{equation} Next, applying again BDG, we have \begin{align*}
      E\left [|\mathbb{M}^{\text{It\^{o}},n}_{s,t}|^{q/2}\right]\lesssim E\left [\left( \int_{s\land \tau_n}^{t\land \tau_n}(M_{s,u})^2|\sigma_u|^2du\right)^{q/4}\right] \leq n^{q/2}E\left [\left( \int_{s\land \tau_n}^{t\land \tau_n}(M_{s,u})^2du\right)^{q/4}\right].
  \end{align*} But using \eqref{hoelderbounds}, for all $\alpha < 1/2$ we have \begin{equation*}
      E\left [\left( \int_{s\land \tau_n}^{t\land \tau_n}(M_{s,u})^2du\right)^{q/4}\right] \leq \frac{1}{2\alpha+1}|t-s|^{q/4(2\alpha+1)}E\left [\left(\sup_{0\leq s \leq t \leq T}\frac{|M^n_t-M^n_s|}{|t-s|^{\alpha}}\right)^{q/2}\right].
  \end{equation*} Since $\alpha < 1/2$, we can conclude that \begin{equation}
      E\left [|\mathbb{M}^{\text{It\^{o}},n}_{s,t}|^{q/2}\right] \leq C_2|t-s|^{q\alpha}.\label{Hoelder2}
  \end{equation} Finally, using the It\^{o}-Stratonovich correction, we find \begin{equation}
      E\left [|\mathbb{M}^{\text{Strat},n}_{s,t}|^{q/2}\right] \lesssim C_2|t-s|^{q\alpha} + E\left[\left (\int_{s\land \tau_n}^{t\land \tau_n}|\sigma_u|^2du\right )^{q/2}\right] \leq C_3|t-s|^{q\alpha}.\label{Hoelder3} 
  \end{equation} Combining \eqref{Hoelder1},\eqref{Hoelder2} and \eqref{Hoelder3}, we have for all $\alpha<1/2$ \begin{equation*}
      \Vert M^n_{s,t}\Vert_{L^q} \leq C|t-s|^{\alpha},\quad  \Vert \mathbb{M}^{\text{Strat},n}_{s,t}\Vert_{L^{q/2}} \leq C|t-s|^{2\alpha}, \quad \Vert \mathbb{M}^{\text{It\^{o}},n}_{s,t}\Vert_{L^{q/2}} \leq C|t-s|^{2\alpha}, 
  \end{equation*} for all $q \geq 2$ and some constant $C$. Applying the Kolmogorov criterion for rough paths, see \cite[Theorem 3.1]{friz2020course}, it follows that for all $\alpha \in (1/3,1/2)$, there exist random variables $K_{\alpha}^n \in L^q$, $\mathbb{K}_{\alpha}^{\text{It\^{o}},n} \in L^{q/2}$ and $\mathbb{K}_{\alpha}^{\text{Strat},n} \in L^{q/2}$ for all $q\geq 2$, such that for all $s,t \in [0,T]$ we have \begin{equation*}
      \frac{|M^n_t-M^n_s|}{|t-s|^{\alpha}} \leq K_{\alpha}^n,\quad  \frac{|\mathbb{M}^{\text{It\^{o}},n}_{s,t}|}{|t-s|^{2\alpha}} \leq \mathbb{K}_{\alpha}^{\text{It\^{o}},n}, \quad \frac{|\mathbb{M}^{\text{Strat},n}_{s,t}|}{|t-s|^{2\alpha}} \leq \mathbb{K}_{\alpha}^{\text{Strat},n}.
  \end{equation*} Finally, for any $\epsilon > 0$, we can use the localization to conclude that for all $\alpha \in (1/3,1/2)$, and $R,n$ large enough, we have\begin{equation*}
      P\left [\sup_{0\leq s\leq t \leq T}\frac{|M_t-M_s|}{|t-s|^{\alpha}}>R\right ] \leq \frac{1}{R}E\left[K^n_{\alpha}\right ] + P[\tau_n > T] \leq \epsilon.
  \end{equation*} In the same way we can show that \begin{equation*}
      P\left [\sup_{0\leq s\leq t \leq T}\frac{|\mathbb{M}^{\text{It\^{o}}}_{s,t}|}{|t-s|^{2\alpha}}>R\right ] \leq \epsilon \text{ and } P\left [\sup_{0\leq s\leq t \leq T}\frac{|\mathbb{M}^{\text{Strat}}_{s,t}|}{|t-s|^{2\alpha}}>R\right ] \leq \epsilon.
  \end{equation*} Therefore, by definition of the spaces $\mathscr{C}^{\alpha}$ and $\mathscr{C}^{\alpha}_g$, the claim follows.
\end{proof}
\subsection{Rough paths with Lipschitz brackets}\label{Lipbracketappendix} In this section we discuss some details about rough paths with Lipschitz brackets, which were introduced in Section \ref{Presection}, Definition \ref{defLipschitzbrackets}. 
For $V=\mathbb{R}^d$, the correspondence from Lemma \ref{bijectionlemma} for the space $\mathscr{C}^{\alpha,1+}$ reads  \begin{equation*}
  \mathscr{C}^{\alpha,1+} ([0, T], \mathbb{R}^d) \longleftrightarrow
  \mathscr{C}_g^{\alpha} ([0, T], \mathbb{R}^d) \oplus\mathrm{Lip}_0([0, T],
  \mathbb{S}^d_+),\label{Bijection2}
\end{equation*} where $\mathrm{Lip}_0$ denotes the space of Lipschitz continuous paths starting from $0$. The right hand-side can be seen as a closed subspace of $C^{\alpha} \oplus C_2^{2\alpha}\oplus \mathrm{Lip}_0$, and is thus complete with respect to the norm \begin{equation*}
    \Vert \mathbf{Y}\Vert_{\alpha,1+}:= \Vert Y\Vert_{\alpha} + \Vert \mathbb{Y}^g \Vert_{2\alpha} + \Vert[\mathbf{Y}] \Vert_{\mathrm{Lip}}. \label{RPlipnorm}
\end{equation*} Moreover, we may identify $\mathrm{Lip}_0([0, T],
  \mathbb{S}^d_+)$ with $L_0^{\infty}([0,T],\mathbb{S}^d_+)$, with the Banach isometry $\Vert[\mathbf{Y}] \Vert_{\mathrm{Lip}} = \Vert\mathbf{V}^{\mathbf{Y}} \Vert_{\infty}$. The following lemma sums up the bijective relations and equivalent norms on the space of rough paths with non-decreasing Lipschitz-brackets. \begin{lemma}\label{LipRPBijectionLemma}
      Let $\alpha \in (1 / 3, 1 / 2)$ and $\mathbf{Y} \in
  \mathscr{C}^{\alpha,1+} ([0, T], \mathbb{R}^d)$. Then we have bijections $\mathscr{C}^{\alpha,1+} \longleftrightarrow
  \mathscr{C}_g^{\alpha} \oplus \mathrm{Lip}_0\longleftrightarrow
  \mathscr{C}_g^{\alpha} \oplus L^{\infty}_0$, with \begin{align*}
      \mathbf{Y}_{s,t} = ((\delta Y)_{s,t},\mathbb{Y}_{s,t}) &= ((\delta Y)_{s,t}, \overbrace{\mathrm{Anti}(\mathbb{Y}_{s,t})+(\delta Y)_{s,t}^{\otimes 2}/2}^{=\mathbb{Y}_{s,t}^g}-\frac{1}{2}(\overbrace{(\delta Y)_{s,t}^{\otimes 2}+2\mathrm{Sym}(\mathbb{Y}_{s,t}))}^{=(\delta[\mathbf{Y}])_{s,t}}) \\ & =((\delta Y)_{s,t},\mathbb{Y}_{s,t}^g-\frac{1}{2}\int_s^t \mathbf{V}_u^{\mathbf{Y}}du),
  \end{align*} for $0 \leq t \leq s \leq T$. Moreover, we have \begin{equation*}
      \Vert \mathbf{Y}\Vert_{\alpha,1+}:= \Vert Y\Vert_{\alpha} + \Vert \mathbb{Y}^g \Vert_{2\alpha} + \Vert\delta [\mathbf{Y}] \Vert_{\mathrm{Lip}} =  \Vert Y\Vert_{\alpha} + \Vert \mathbb{Y}^g \Vert_{2\alpha} + \Vert \delta \mathbf{V}^{\mathbf{Y}} \Vert_{\infty},
  \end{equation*} and the embedding $(\mathscr{C}^{\alpha,1+},\Vert \cdot \Vert_{\alpha,1+}) \hookrightarrow (\mathscr{C}^{\alpha},\Vert \cdot \Vert_{\alpha})$ is continuous. \end{lemma}
  \begin{proof}
     The bijection $\mathscr{C}^{\alpha,1+} \longleftrightarrow
  \mathscr{C}_g^{\alpha} \oplus \mathrm{Lip}_0$ follows from Lemma \ref{bijectionlemma}, where we replace $C_0^{2\alpha}$ by $\mathrm{Lip}_0$, which follows directly by definition of the space $\mathscr{C}^{\alpha,1+}$. The second bijection, $\mathscr{C}_g^{\alpha} \oplus \mathrm{Lip}_0 \longleftrightarrow \mathscr{C}^{\alpha}_g \oplus L^{\infty}_0$ follows from the fact that $L^{\infty} \cong \mathrm{Lip}$. This and the fact that $\Vert[\mathbf{Y}] \Vert_{\mathrm{Lip}} = \Vert \mathbf{V}^{\mathbf{Y}} \Vert_{\infty} $  can for instance be found in \cite[Proposition 1.37]{friz2010multidimensional}. Finally, for all $\mathbf{Y} \in \mathscr{C}^{\alpha,1+}$, we have \begin{align*}
      \Vert \mathbf{Y} \Vert_{\alpha} :=  \Vert Y \Vert_{\alpha} +  \Vert \mathbb{Y} \Vert_{2\alpha} &= \Vert Y \Vert_{\alpha} +  \Vert \mathbb{Y}^g-\frac{1}{2}\delta[\mathbf{Y}] \Vert_{2\alpha} \\ & \leq \Vert Y \Vert_{\alpha} +  \Vert \mathbb{Y}^g \Vert_{2\alpha}+C\Vert\delta[\mathbf{Y}] \Vert_{\mathrm{Lip}}  \\ & \leq C \Vert \mathbf{Y}\Vert_{\alpha,1+},
  \end{align*} where the constant $C$ changed from the second to the last line. 
  \end{proof} \\ As discussed in Remark \ref{CorresMG}, the motivation to study the space $\mathscr{C}^{\alpha,1+}$ comes from the fact that the It\^{o} rough path lifts of local martingales, see Proposition \ref{Mglift}, constitute examples of (random) rough paths with non-decreasing Lipschitz brackets. This can also be deduced from the following general result. \begin{lemma}
    Let $\alpha \in (1/3,1/2)$ and $\mathbf{Y}^g = (Y,\mathbb{Y}^g) \in \mathscr{C}_g^{\alpha}([0,T],\mathbb{R}^d)$. For any $\mathbf{v}\in L^{\infty}([0,T],\mathbb{R}^{d \times k})$, set $\mathbf{V}:=\mathbf{v}\mathbf{v}^T $ and \begin{equation*}
        \mathbf{Y}_t:= \left (Y_t,\mathbb{Y}^g_{0,t}-\frac{1}{2}\int_0^t\mathbf{V}_sds\right ).
    \end{equation*} Then $\mathbf{Y}\in \mathscr{C}^{\alpha,1+}([0,T],\mathbb{R}^d)$ and $\mathbf{V}^{\mathbf{Y}}= \mathbf{V}$.
\end{lemma} 
\begin{proof}
    Defining $\mathbb{H}_{s,t} := -\frac{1}{2}\int_s^t\mathbf{V}_udu$, it follows from Lemma \ref{bijectionlemma} that $\mathbf{Y}\in \mathscr{C}^{\alpha}$. Now by Definition \ref{rBrackets}, we have \begin{align*}
        [\mathbf{Y}]_t := Y_{0,t} \otimes Y_{0,t}-2\mathrm{Sym}(\mathbb{Y}_{0,t}) & = Y_{0,t} \otimes Y_{0,t}-2\mathrm{Sym}(\mathbb{Y}^g_{0,t})+\int_0^t\mathbf{V}_sds \\ & = \int_0^t\mathbf{V}_sds,
    \end{align*} where we used that $\mathbb{H}_{0,t}$ is symmetric, and $\mathbf{V}:= \mathbf{v}\mathbf{v}^T \in \mathbb{S}^d_+$, and the fact that $\mathbf{Y}^g$ is a geometric rough path. But since $\mathbf{V}\in L^{\infty}([0,T],\mathbb{S}^d_+)$, it readily follows that  $\int_0^t\mathbf{V}_sds$ is Lipschitz with \begin{equation*}
        \mathbf{V}^{\mathbf{Y}}_t:= \frac{d[\mathbf{Y}]_t}{dt}=\mathbf{V}_t\in \mathbb{S}^d_+,
    \end{equation*} which finishes the proof.
\end{proof}
\begin{remark}
    In the case of scalar rough path spaces, that is $d=1$, it is not hard to
see that every $\alpha$-H{\"o}lder continuous path $Y: [0, T]
\rightarrow \mathbb{R}$ has a trivial geometric lift given by $\mathbb{Y}_{s,
t} \assign \frac{1}{2} (Y_{s, t})^2$. Thus, one can simply identify $\mathscr{C}^{\alpha}_g ([0, T], \mathbb{R})$ with
$C^{\alpha} ([0, T], \mathbb{R})$, and in view of the one-to-one correspondences above, we have $ \mathscr{C}^{\alpha,1+} \longleftrightarrow
  C^{\alpha}([0,T],\mathbb{R})\oplus\mathrm{Lip}_0([0,T],\mathbb{R}_+)\longleftrightarrow
  C^{\alpha}([0,T],\mathbb{R})\oplus L_0^{\infty}([0,T],\mathbb{R}_+)$, and the norms can be reduced to \begin{equation*}
       \Vert \mathbf{Y}\Vert_{\alpha,1+}:= \Vert Y\Vert_{\alpha} + \Vert[\mathbf{Y}] \Vert_{\mathrm{Lip}} =  \Vert Y\Vert_{\alpha} +  \Vert \mathbf{V}^{\mathbf{Y}} \Vert_{\infty}.
  \end{equation*}
\end{remark}

\subsection{Proof of Theorem \ref{consistencycoro}}\label{RSDEAPPendix} In this section we prove a multivariate generalization of Theorem \ref{consistencycoro}. Let $\mathbf{Y}\in \mathscr{C}^{\alpha,1+}([0,T],\mathbb{R}^d)$ for some $d\geq 1$ and $\alpha \in (1/3,1/2)$, and recall the notation in Definition \ref{defLipschitzbrackets}\begin{equation*}
    \left ([\mathbf{Y}]_t,\frac{d[\mathbf{Y}]_t}{dt},\sqrt{\frac{d[\mathbf{Y}]_t}{dt}} \right )=: \left ( \int_0^t\mathbf{V}^{\mathbf{Y}}_sds,\mathbf{V}^{\mathbf{Y}}_t,\mathbf{v}^{\mathbf{Y}}_t\right ).\end{equation*}
 Let $d_X \geq 1$ and consider two functions $f:[0,T]\times \mathbb{R}^{d_X} \longrightarrow \mathbb{R}^{d_X\times d}$ and $g:[0,T]\times \mathbb{R}^{d_X} \longrightarrow \mathbb{R}^{d_X\times d}$ in $C_b^3$, and let $B$ be a $d$-dimensional Brownian motion. As motivated in Section \ref{multivariatesection}, the $d_X$-dimensional version of \eqref{RSDErep4} is given by \begin{equation}
    \begin{aligned}
        X_t^{t, x, \tmmathbf{Y}} & =  x \in \mathbb{R}^{d_X}, \quad dX^{t, x, \tmmathbf{Y}}_s  = f( {s, X_s^{t, x, \tmmathbf{Y}}}  )d\mathbf{Y}_s + g(s, X^{t, x, \tmmathbf{Y}}_s)
  \mathbf{v}^{\mathbf{Y}}_s dB_s,  \quad t<s\le T.
  \label{HRSDEA}
    \end{aligned}
\end{equation}

Now consider the $d$-dimensional local martingale \begin{equation}
    M^{\mathbf{Y}}_t = \int_0^t\mathbf{v}^{\mathbf{Y}}_sdB_s = \left ( \sum_{j=1}^d\int_0^t(\mathbf{v}^{\mathbf{Y}}_s)^{ij}dB^j_s\right )_{i=1}^d\label{MGG}.
\end{equation} The general version of the joint-lift of $\begin{pmatrix}
    M^{\mathbf{Y}} \\ Y
\end{pmatrix}\in \mathbb{R}^{2d}$ from Definition \ref{jointlift} is given by
\begin{equation}
  Z_{s,t}^{\mathbf{Y}}:=\begin{pmatrix}
    M_{s,t}^{\mathbf{Y}} \\ Y_{s,t}
\end{pmatrix}, \quad \mathbb{Z}^{\mathbf{Y}}_{s,t} \assign \left(\begin{array}{cc}
    \int_s^t M^{\mathbf{Y}}_{s,r}\otimes dM^{\mathbf{Y}}_r & \int_s^t M^{\mathbf{Y}}_{s,r}\otimes dY_r\\
    \int_s^t Y_{s,r}\otimes dM^{\mathbf{Y}}_r & \mathbb{Y}_{s,t}
  \end{array}\right), \quad 0\leq s\leq t \leq T, \label{lift_appendix}
\end{equation} where the first entry is the (canonical) It{\^o} rough path lift of the local
martingale $M^{\mathbf{Y}}$, see Proposition \ref{Mglift}, $\int_s^t Y_{s,r} \otimes dM^{\mathbf{Y}}_r$ is a well-defined It{\^o}
integral, and set $\int_s^t  M^{\mathbf{Y}}_{s,r}\otimes dY_r : = M^{\mathbf{Y}}_{s,t}\otimes Y_{s,t} - \int_s^tY_{s,r}\otimes dM^{\mathbf{Y}}_r$, imposing integration by parts. A general version of Theorem \ref{consistencycoro} can be stated as follows. \begin{theorem}
  \label{consistencycoroA}Let $\alpha \in (1 / 3, 1 / 2)$ and $\tmmathbf{Y} \in
  \mathscr{C}^{\alpha,1+} ([0, T], \mathbb{R}^d)$, and $M^{\mathbf{Y}}$ given in \eqref{MGG}. Then $\tmmathbf{Z}^{\mathbf{Y}} (\omega) = (Z^{\mathbf{Y}}(\omega),\mathbb{Z}^{\mathbf{Y}}(\omega))$ defines an
  $\alpha'$-H{\"o}lder rough path for any $\alpha' \in (1 / 3, \alpha)$, and for $f,g \in C_b^3$, there exists a unique
  solution to the rough differential equation
  \begin{equation}
    X^{t,x,\mathbf{Y}}_t=x  \in \mathbb{R}^{d_X}, \quad dX^{t, x, \tmmathbf{Y}}_s (\omega) = (g, f) (s, X_s^{t,
    x, \tmmathbf{Y}} (\omega)) d \tmmathbf{Z  }^{\mathbf{Y}}_s (\omega), \quad t <s \leq T, \label{RDEStrA}
  \end{equation}
  for almost every $\omega$, and $X^{t,x,\mathbf{Y}}$ defines a time-inhomogeneous Markov process. Moreover, let $W$ be another $d$-dimensional Brownian motion independent of $B$, and assume $v$ is $\mathbb{R}^{d\times d}-$valued and $(\mathcal{F}^W_t)-$adapted such that Assumption \ref{ass1} holds true. For $I_t = \int_0^tv_sdW_s$, if we choose $\tmmathbf{Y} =
  \tmmathbf{I} (\omega)$ the It\^{o} rough path lift of $I$, then, for a.e. $\omega$ we have \begin{equation*}
      \mathrm{Law}\left (\left. X^{t,x} \right |\mathcal{F}^W_T\lor \mathcal{F}_t^B\right)(\omega) = \mathrm{Law}\left (\left. X^{t,x,\mathbf{I}}\right |\mathcal{F}^W_T \lor \mathcal{F}_t^B\right)(\omega) =\mathrm{Law}(\left. X^{t,x,\mathbf{Y}})\right |_{\mathbf{Y=\mathbf{I}(\omega)}}, 
  \end{equation*}where $X^{t,x}$ is the unique strong solution to \begin{equation}
    \begin{aligned}
      X_t^{t, x} = x, \quad dX_s^{t, x}  = f(s,
  X_s^{t, x}) v_s dW_s + g(s, X_s^{t, x}) v_s dB_s , \quad t< s \leq T.\label{dynamicsA}
    \end{aligned}
\end{equation} Finally, if $v\in \mathbb{S}^{d}_+$, then we have indistinguishability $X^{t,x}(\omega) = X^{t,x,\mathbf{I}(\omega)}$ for a.e. $\omega$. 
\end{theorem} We split the proof of Theorem \ref{consistencycoroA} into the following three lemmas.

\begin{lemma}
  \label{roughIto}Consider a $d$-dimensional local martingale of the form \begin{equation*}
      M_t=\int_0^t\sigma_sdB_s,
  \end{equation*}where $\sigma$ is progressively measurable and in $L^{\infty}([0,T],\mathbb{R}^{d\times d})$ almost surely. Then \begin{enumerate}
      \item The lift $\tmmathbf{Z}^{\mathbf{Y}} (\omega)$, similarly constructed as in \eqref{lift_appendix} with $M$, defines an
  $\alpha'$-H{\"o}lder rough path for any $\alpha' \in (1 / 3, \alpha)$, that is $\mathbf{Z}^{\mathbf{Y}} \in \mathscr{C}^{\alpha'}([0,T],\mathbb{R}^{2d})$.
  \item Consider a pathwise controlled rough path $(A (\omega), A'
  (\omega)) \in \mathscr{D}_{\mathbf{Z}^{\mathbf{Y}} (w) }^{2 \alpha'}([0,T],\mathbb{R}^{d_X \times 2d})$, see \cite[Chapter 4]{friz2020course} for the definition of the space of controlled rough paths $\mathscr{D}$. Then, the rough integral \begin{equation*}
      \mathcal{J} (T, \omega, \tmmathbf{Y})  \assign \int_0^T A_s (\omega)d\mathbf{Z}_s^{\mathbf{Y}}(\omega) = \lim_{| \pi | \rightarrow 0}\sum_{[u, v] \in \pi} A_u (\omega) Z_{u, v} (\omega) + A'_u (\omega)
     \mathbb{Z}_{u, v} (\omega)
  \end{equation*}
  exists on a set $\Omega_0 \subseteq \Omega$ with full measure.
  \end{enumerate}
   
\end{lemma}

\begin{proof}
  First we show that \tmtextbf{}$\tmmathbf{Z}^{\mathbf{Y}}(\omega)$ indeed defines a
  (random) $\alpha'$-H{\"o}lder rough path for any $\alpha' \in (1 / 3,
  \alpha)$. Notice that the proof uses similar techniques as in Proposition \ref{Mglift}. Define the stopping time \begin{equation*}
      \tau _n \assign \inf \left\{ t \geq 0 :
  \Vert \sigma\Vert_{\infty;[0,t]} \geq n \right\}.
  \end{equation*} For any $q \geq 2,$ we
  can apply the Burkholder-Davis-Gundy inequality, see \cite[Theorem 18.7]{kallenberg1997foundations}, to find \begin{equation*}
      E \left[ \left| \int_s^t {\sigma^{\tau_n}_u}^{}  dB_u \right|^q
     \right]^{1 / q} \leq C_q n  (t - s)^{1 / 2}  \infixand | Y_{s, t}
     | \leq C | t - s |^{\alpha},
  \end{equation*} for some constants $C_q$ and C. Moreover, we clearly have $\| \mathbb{Y}_{s,
  t} \|_{L^{q / 2}} \leq C_1 | t - s |^{2 \alpha}$ for all $q \geq 2$. From Proposition \ref{Mglift}, we already know that \begin{equation*}
      \left\| \int_s^t M^{\tau_n}_{s, r} \otimes dM^{\tau_n}_r
     \right\|_{L^{q / 2}} \leq C_2|t-s|^{2\alpha},
  \end{equation*} Similarly, we find \begin{equation*}
       \left\| \int_s^t Y_{s, u} \otimes {dM^{\tau_n}_u}  \right\|_{L^{q
     / 2}} \lesssim \| Y \|_{\alpha }  (t - s)^{\alpha + 1/2}
  \end{equation*} and \begin{equation*}
      \left\| \int_s^t  M^{\tau_n}_{s, u} \otimes dY_u  \right\|_{L^{q /
     2}} \leq \|  M^{\tau_n}_{s, t}\otimes  Y_{s, t} \|_{L^{q / 2}} +
     \left\| \int_s^t Y_{s, u} \otimes {dM^{\tau_n}_u}  \right\|_{L^{q
     / 2}} \lesssim (t - s)^{2 \alpha}.
  \end{equation*} Applying Kolomogorov for rough paths, that is \cite[Theorem 3.1]{friz2020course}, it
  follows that $\tmmathbf{Z}^{\mathbf{Y}}_{\tau_n \wedge \cdot} (\omega)$ indeed defines
  a $\alpha'$-H{\"o}lder rough path for any $\alpha' \in (1 / 3, \alpha)$. Using exactly the same localization argument as in Proposition \ref{Mglift}, we
  can conlude that $\mathbf{Z}^{\mathbf{Y}}$ defines a $\alpha'$-H{\"o}lder rough path. By
  \cite[Theorem 4.10]{friz2020course}, the rough integral $\mathcal{J} (T, \omega,
  \tmmathbf{Y})$ is therefore well-defined on a set with full measure.
\end{proof} 
\begin{remark} We remark that it is actually possible to choose $\alpha'=\alpha$ in (i) of Lemma \ref{roughIto}, which however would require a more refined version of the Kolmogorov criterion with inhomogeneous exponents. Providing such a result here is outside of the scope the paper, and since the need of an $\alpha'<\alpha$ does not affect our main results, we leave this as a remark here.
\end{remark}
\begin{lemma}\label{consistenyALemma} Let $M$ be the $d$-dimensional local martingale from Lemma \ref{roughIto}, and let
  $N$ be another $d$-dimensional local martingale, such that $[M, N] = 0$. Consider two paths $A:[0,T]\rightarrow \mathbb{R}^{d_X\times 2d}$ and $A':[0,T] \rightarrow \mathbb{R}^{d_X} \otimes \mathbb{R}^{2d\times 2d}$. If $(A,A')$ is an adapted
  and continuous controlled rough path $(A, A')\in \mathscr{D}_{\mathbf{Z}^{\mathbf{Y}}(\omega)}^{\alpha'}([0,T],\mathbb{R}^{d_X \times 2d})$ for any $\alpha' \in (1/3,\alpha)$, then it holds for almost every $\omega$ that
  \begin{equation}
    \mathcal{J} (T, \omega, \tmmathbf{Y})|_{\tmmathbf{Y} = \tmmathbf{N}
    (\omega)} = \left (\int_0^TA_sdZ^{\mathbf{N}}_s\right )(\omega) =  \left (\int_0^TA^{1}_s dN_s \right)(\omega)+\left (\int_0^TA^{2}_s dM_s\right )(\omega)\in \mathbb{R}^{d_X}, \label{genuineRDE}
  \end{equation}
  where $A^1,A^2 \in \mathbb{R}^{d_X \times d}$ such that $A=(A^1,A^2)$, and $\mathcal{J}$ is defined in Lemma \ref{roughIto}. The integrals on the right-hand side are It{\^o} integrals, and
  $\tmmathbf{N}$ is the It{\^o} rough path lift of $N$.
\end{lemma}\begin{remark}\label{dimensionrmk}
    Let us quickly describe how to understand the two different notions of integration in \eqref{genuineRDE} in the sense of dimension. First, since $Z:=Z^{\mathbf{N}}=\begin{pmatrix}
    N \\ M
\end{pmatrix}$ is a $2d$-dimensional local martingale, the stochastic integration in \eqref{genuineRDE} can be understood for each component $i\in \{1,\dots,d_X\}$ as \begin{equation*}
        \left(\int_0^TA_sdZ_s\right)^{i}=\sum_{j=1}^{2d}\int_0^TA^{i,j}_sdZ^j_s = \sum_{j=1}^d\left (\int_0^TA^{1,ij}_sdN^{j}_s+\int_0^TA^{2,ij}_sdM^{j}_s\right ).
    \end{equation*} On the other hand, by the definition of the rough integral $\mathcal{J}$ in (ii) of Lemma \ref{roughIto}, we encounter the second order term $A'\mathbb{Z}$, where $A'$ is the \textit{Gubinelli derivative}\footnote{c.f. \cite[Definition 4.6]{friz2020course}.}. Here $A'$ takes values in $\mathbb{R}^{d_X}\otimes \mathbb{R}^{2d\times 2d} \cong \mathbb{R}^{d_X\times (2d)^2}$, and since $\mathbb{Z}:=\mathbb{Z}^{\mathbf{N}}\in \mathbb{R}^{2d\times 2d} \cong \mathbb{R}^{(2d)^2}$, the product $A'\mathbb{Z}$ lies again in $\mathbb{R}^{d_X}$, such that the left hand side of \eqref{genuineRDE} is also $\mathbb{R}^{d_X}$-valued. In the simple case $d_X=d=1$, which is of main interested in this paper, we simply have $A:[0,T] \rightarrow \mathbb{R}^2$ and $A':[0,T] \rightarrow \mathbb{R}^{2\times 2}$, and the meaning of the terms above should be clear.
\end{remark}\begin{proof} Set $Z=Z^{\mathbf{N}}$ and $\mathbb{Z}=\mathbb{Z}^{\mathbf{N}}$. Since $A$ is continuous and adapted, the stochastic integral on the right-hand side of \eqref{genuineRDE} is
  well-defined and it is given as limit in probability \begin{equation*}
      \int_{0 }^T A_s dZ_s = (p) \lim_{| \pi | \rightarrow 0} \sum_{[u,
     v] \subset \pi} A_u Z_{u, v} .
  \end{equation*}
  Moreover, there exists a subsequence of subdivisions $\pi^n$, such that the
  convergence holds almost surely. By definition of the rough integral,
  we find \begin{equation*}
      \mathcal{J} (T, \omega, \tmmathbf{Y}) |_{\tmmathbf{Y} = \tmmathbf{N}
     (\omega)}= \lim_{| \pi | \rightarrow 0} \sum_{[u, v]
     \in \pi} A_u (\omega) Z_{u, v} (\omega) + A'_u (\omega) \mathbb{Z}_{u, v}
     (\omega) . 
  \end{equation*}
  Along the subsequence $\pi^n$ we find \begin{equation*}
      \int_0^T A_s d\mathbf{Z}^{\mathbf{N}}_s - \int_0^T A_s dZ_s = \lim_{n
     \rightarrow \infty} \sum_{[u, v] \subset \pi^n} A'_u \mathbb{Z}_{u, v} .
  \end{equation*} Assume for the moment that the processes $A'$, N and M are bounded by some
  $K$. Changing the notation to $\pi^n = \{ 0 = t_0 < t_1 < \cdots < t_n = T
  \}$ and for any $r\in \{1,\dots,d_X\}$, denote by $A'^{(r)}$ the $r$-th row of the Gubinelli derivative $A'\in \mathbb{R}^{d_X \times (2d)^2}$, see also Remark \ref{dimensionrmk}. Then \begin{align*}
      E \left[ \left( \sum_{[u, v] \subset \pi^n} A'^{(r)}_u \cdot \mathbb{Z}_{u, v}  \right)^2\right] & = \sum_{j=0}^{n-1}E[(A_{t_j}'^{(r)}\mathbb{Z}_{t_j,t_{j+1}})^2] \\ & \qquad \qquad + 2 \sum_{i < j}E[(A_{t_j}'^{(r)}\mathbb{Z}_{t_j,t_{j+1}})(A_{t_i}'^{(r)}\mathbb{Z}_{t_i,t_{i+1}})].
  \end{align*} Now by definition of $\mathbb{Z}=\mathbb{Z}^{\mathbf{N}}$, see \eqref{lift_appendix}, we have \begin{equation*}
      \mathbb{Z} = \left (\begin{array}{cc}
    \int_s^t M_{s,r}\otimes dM_r & M_{s,t} \otimes N_{s,t}-\int_s^t N_{s,r}\otimes dM_r\\
    \int_s^t N_{s,r}\otimes dM_r & \int_s^tN_{s,r}\otimes dN_r
  \end{array}\right) = \left (\begin{array}{cc}
    \int_s^t M_{s,r}\otimes dM_r & \int_s^t M_{s,r}\otimes dN_r\\
    \int_s^t N_{s,r}\otimes dM_r & \int_s^tN_{s,r}\otimes dN_r
  \end{array}\right),
  \end{equation*} where the second equality follows from It\^{o}'s formula and the fact that $[M,N] = 0$. Since $A'$ is adapted, we can apply the tower-property to see that for any $t_i < t_{i + 1} \leq t_j < t_{j + 1}$ we have
  \begin{equation*}
      E[(A_{t_j}'^{(r)}\mathbb{Z}_{t_j,t_{j+1}})(A_{t_i}'^{(r)}\mathbb{Z}_{t_i,t_{i+1}})] = E[(A_{t_i}'^{(r)}\mathbb{Z}_{t_i,t_{i+1}})A_{t_j}'^{(r)}E[\mathbb{Z}_{t_j,t_{j+1}}|\mathcal{F}_{t_j}]] = 0,
  \end{equation*} where we use the martingale property for all the stochastic integral entries of $\mathbb{Z}$, to see that $E[\mathbb{Z}_{t_j,t_{j+1}}|\mathcal{F}_{t_j}] = 0$.Thus, we have \begin{equation*}
      E\left[ \left( \sum_{[u, v] \subset \pi^n} A'^{(r)}_u \mathbb{Z}_{u, v} . \right)^2\right ]
     \leq  K^2 \sum_{l=1}^{(2d)^2} \left( \sum_{j = 0}^{n-1} E \left[ ({\mathbb{Z 
     }_{t_j, t_{j + 1}}^{(l)}})^2 \right] \right).
  \end{equation*} Using standard stochastic integral properties, we have 
  \begin{align*}
    \sum_{j = 0}^{n - 1} \left( \int_{t_j}^{t_{j + 1}}  N_{t_j, u} \otimes dM_u \right) ^2 & \leq  \left( \sum_{j = 0}^{n - 1}
    \int_{t_j}^{t_{j + 1}} N_u \otimes dM_u - N_{t_j}\otimes  (M_{t_{j + 1}} - M_{t_j})
    \right)^2 \\
    & \leq  \left( \int_{0 }^T N_u \otimes dM_u - \sum_{j = 0}^{n-1}N_{t_j} \otimes (M_{t_{j +
    1}} - M_{t_j})  \right)^2  \\ &\xrightarrow{n \to \infty}  0 \quad \text{in }L^1,
  \end{align*}
  since both $M$ and $N$ are bounded. The same holds true when we exchange the
  positions of $M$ and $N$, and also if $M = N$. Therefore, under the boundedness
  assumption, we have \begin{equation*}
      E\left[ \left( \sum_{[u, v] \subset \pi^n} A'^{(r)}_u \mathbb{Z}_{u, v} . \right)^2\right ]
     \xrightarrow{n \to \infty} 0,
  \end{equation*} and in this case, the stochastic and rough integral coincide almost surely. \\ Define now $\tau^K = \inf \left\{ t \geq 0 : \| M_t \| \geq K \infixor | |
  \nobracket A_t' \| \geq K \infixor \| N_t \| \geq K \right\} $. Since all
  processes are continuous, we clearly have $\tau^K \rightarrow \infty$ as $K
  \rightarrow \infty .$ Therefore, we find \begin{align*}
      P\left[ \left| \sum_j A'^{(r)}_{t_j} \mathbb{Z}_{t_j, t_{j + 1}}
     \right| > \epsilon \right] & \leq P\left[ \left| \sum_{j ; t_{j +
     1} < \tau^K} A'^{(r)}_{t_j} \mathbb{Z}_{t_j, t_{j + 1}} \right| > \epsilon
     \right] + P [\tau^K \leq T] \\ &  \leq \frac{1}{\epsilon^2} E
     \left [\left| \sum_{j ; t_{j + 1} < \tau^K} A'^{(r)}_{t_j} \mathbb{Z}_{t_j, t_{j + 1}}
     \right|^{ 2}\right ] + P [\tau^K \leq T] \\ & \xrightarrow{n\to \infty} \Delta (K), 
  \end{align*} where $\Delta (K) \rightarrow 0$ as $K \rightarrow \infty .$ This shows that
  $\sum_{_{[u, v] \subset \pi^n}} A'_u \mathbb{Z}_{u, v} \rightarrow 0$ in
  probability, hence we can find another subsequence $\pi^{n_k}$, such that
  the convergence holds almost surely. It follows that almost surely \begin{equation*}
      \int_0^T A_s d \tmmathbf{Z}^{\mathbf{N}}_s = \int_0^T A_s dZ_s.
  \end{equation*}
\end{proof} \\ To establish a Markovian structure for the process $X^{t,x,\mathbf{Y}}$, we need the following lemma, which tells us that $X^{t,x,\mathbf{Y}}$ is the limit of time-inhomogeneous Markov processes. This is the key ingredient to establish a Feynman-Kac type of result for the dynamics $X^{t,x,\mathbf{Y}}$ in Section \ref{RPDEsection}.
\begin{lemma}
    \label{MarkovProperty}
    Let $\alpha \in (1/3,1/2)$ and $\mathbf{Y} \in
  \mathscr{C}^{\alpha,1+} ([0, T], \mathbb{R}^d)$, and consider $M^{\mathbf{Y}}$ as defined in \eqref{MGG}. Let $\mathbf{Y}^{\epsilon}$ be the rough path lift of a piecewise linear approximation\footnote{See \cite[Chapter 5.2]{friz2010multidimensional}  for instance.} $Y^{\epsilon}$ of $Y$, such that $\mathbf{Y}^{\epsilon} \rightharpoonup \mathbf{Y}$, see Proposition \ref{weakconvpropo}. For $(t,x) \in [0,T] \times \mathbb{R}^{d_X}$ assume that $X^{t,x,\mathbf{Y}}$, resp. $X^{t,x,\mathbf{Y}^{\epsilon}}$ is the unique solution to \eqref{RDEStrA} with respect to $\mathbf{Z}^{\mathbf{Y}}$, resp. $\mathbf{Z}^{\mathbf{Y}^{\epsilon}}$. Then we have \begin{equation*}
         X^{t,x,\mathbf{Y}^{\epsilon}} \xrightarrow{\epsilon \to 0} X^{t,x,\mathbf{Y}} \text{ ucp}.
  \end{equation*} Moreover, we have $X^{t,x,\mathbf{Y}^{\epsilon}} = X^{t,x,\epsilon}$ almost surely, where $X^{t,x,\epsilon}$ denotes the unique strong solution to the SDE \begin{align*}
      X^{t,x}_t=x, \quad dX_s^{t, x, \epsilon} = g(s, X^{t, x,
    \epsilon}_s) \mathbf{v}_s^{\mathbf{Y}}dB_s + \left( f_0 (s,
    X_s^{t, x, \epsilon}) + f \left( {s, X_s^{t, x,
    \epsilon}}  \right) \dot{Y}_s^{\epsilon} \right) ds,
  \end{align*}\color{black} where $f_0:[0,T] \times \mathbb{R}^{d_X} \rightarrow \mathbb{R}^{d_X}$ is given by $f_0^i(t,x)= -\frac{1}{2}\sum_{j,k,l}\partial_{x_k}f^{ij}(t,x)f^{kl}(t,x)(\mathbf{V}_t^{\mathbf{Y}})^{lj}.$
\end{lemma} 
\begin{proof} For any $\alpha'<\alpha$ , assume for the moment that the rough path norms of $\mathbf{Z}^{\mathbf{Y}}$ and $\mathbf{Z}^{\mathbf{Y}^{\epsilon}}$ are bounded by some constant $K$. Then we can apply standard RDE estimates, e.g. \cite[Theorem 8.5]{friz2020course}, to see that \begin{equation}
    \sup_{t\leq s \leq T}\left |X_s^{t,x,\mathbf{Y}^{\epsilon}}-X_s^{t,x,\mathbf{Y}}\right | \leq C\varrho_{\alpha'}(\mathbf{Z}^{\mathbf{Y}^{\epsilon}},\mathbf{Z}^{\mathbf{Y}}), \label{toestimate}
\end{equation} where $C=C(K,\alpha',\alpha,f,g)$ is constant. Now since $M^{\mathbf{Y}^{\epsilon}}=M^{\mathbf{Y}}$, for any $q \geq 1$ and $t\leq s \leq s'\leq T$ we have \begin{equation*}
    E[|Z_{s,s'}^{\mathbf{Y}^{\epsilon}}-Z_{s,s'}^{\mathbf{Y}}|^q]^{1/q} = |Y^{\epsilon}_{s,s'}-Y_{s,s'}| \leq (s'-s)^{\alpha'} \Vert \mathbf{Y}^{\epsilon}-\mathbf{Y}\Vert_{\alpha',1+}.
\end{equation*} In order to estimate $E[|\mathbb{Z}_{s,s'}^{\mathbf{Y}^{\epsilon}}-\mathbb{Z}_{s,s'}^{\mathbf{Y}}|^{q/2}]$, we use similar techniques as in the proof of Lemma \ref{roughIto}, by applying BDG-inequalities \begin{equation*}
     E\left [\left |\int_s^{s'}(Y_{s,r}^{\epsilon}-Y_{s,r})\otimes dM^{\mathbf{Y}}_r\right |^{q/2} \right ]^{2/q} \leq C_{q/2}(s'-s)^{2\alpha'}\Vert V^{\mathbf{Y}}\Vert_{\infty;[0,T]}\Vert \mathbf{Y}^{\epsilon}-\mathbf{Y}\Vert_{\alpha',1+}.
\end{equation*} Applying similar estimates for $\left |\int_s^{s'}M_{s,r}^{\mathbf{Y}}\otimes dY^{\epsilon}_r-\int_s^{s'}M_{s,r}^{\mathbf{Y}} \otimes dY_r\right |^{q/2}$, we find  \begin{equation*}
    E\left [\left |\mathbb{Z}_{s,s'}^{\mathbf{Y}^{\epsilon}}-\mathbb{Z}_{s,s'}^{\mathbf{Y}}\right |^{q/2} \right ]^{2/q} \leq \tilde{C}_{q/2}(s'-s)^{2\alpha'}\Vert \mathbf{Y}^{\epsilon}-\mathbf{Y}\Vert_{\alpha',1+} (3\Vert V^{\mathbf{Y}}\Vert_{\infty;[0,T]}+1)
\end{equation*} Applying a Kolmogorov criterion for rough path distances, see \cite[Theorem 3.3]{friz2020course}, it follows that \begin{equation*}
     E\left [\sup_{t\leq s \leq T}\left |X_s^{t,x,\mathbf{Y}^{\epsilon}}-X_s^{t,x,\mathbf{Y}}\right |^q\right]^{1/q} \lesssim E [\varrho_{\alpha'}
    (\tmmathbf{Z}^{\mathbf{Y}^{\epsilon}}, \tmmathbf{Z}^{\mathbf{Y}})^q] ^{1/q}\lesssim \Vert \mathbf{Y}^{\epsilon}-\mathbf{Y} \Vert_{\alpha',1+}\xrightarrow{\epsilon \to 0} 0, \quad \forall q\geq 1.
  \end{equation*} Now for the general case, that is where the rough path norms of $\mathbf{Z}^{\mathbf{Y}}$ and $\mathbf{Z}^{\mathbf{Y}^{\epsilon}}$ might be unbounded, we apply once again a standard localization argument, to conclude that the convergence holds in probability.  \\ Finally, the claim $X^{t,x,\mathbf{Y}^{\epsilon}} = X^{t,x,\epsilon}$ follows with a standard consistency argument, similar as in the proof of Lemma \ref{roughIto} but with Riemann-Stieltjes integration instead of stochastic integration.
\end{proof}

\begin{proof}{\textbf{of Theorem \ref{consistencycoroA}}}
  The existence and uniqueness of a solution to \eqref{RDEStrA} follows from
  the main well-posedness result for rough differential equations, see \cite[Theorem 8.4]{friz2020course},
  where we only need to define the augmented path $\hat{X}_s = (s, X_s)$
  and the augmented lift $\hat{\mathbf{Z}}^{\mathbf{Y}}$ by including the missing
  cross-integrals as Riemann-Stieltjes integrals. Let us show that $X^\mathbf{Y}$ is a $(\mathcal{F}^B_t)-$Markov process, that is for all bounded and measurable functions $k$, we have \begin{equation}\label{eq:MarkovPropertyAppendix}
      E[k(X^{\mathbf{Y}}_{t+h})|\mathcal{F}_t^B] = E[k(X^{\mathbf{Y}}_{t+h})|X^{\mathbf{Y}}_t], \quad  0\leq h \leq T-t,
  \end{equation} where $X^{\mathbf{Y}}=X^{0,x_0,\mathbf{Y}}$. First, we can notice that $\mathbf{Z}^{\mathbf{Y}}|_{\Delta_{[t,s]}}= (\mathbf{Z}^{\mathbf{Y}}_{u,w}:t\leq u \leq w \leq s)$, see \eqref{lift_appendix} for the definition of $\mathbf{Z}^{\mathbf{Y}}$, is measurable with respect to $\mathcal{F}^B_{t,s}:=\sigma(B_w-B_u: t\leq u \leq w \leq s)$. Indeed, by definition of the stochastic integral this holds true for the increments of the Gaussian martingales $M^{\mathbf{Y}}=\int \mathbf{v}^{\mathbf{Y}}dB$ and $\int (\delta Y)\mathbf{v}^{\mathbf{Y}}dB$, and therefore also for the two remaining  stochastic terms $(\delta Y)(\delta M^{\mathbf{Y}})-\int (\delta Y)\mathbf{v}^{\mathbf{Y}}dB$ and  $\int (\delta M^{\mathbf{Y}})\mathbf{v}^{\mathbf{Y}}dB$. Now define the (deterministic) solution map to the RDE \eqref{RDEStrA} by $\Psi_{s \leftarrow t}^{\mathbf{Z^{\mathbf{Y}}}|_{[t,s]}}(x):=X_s^{t,x,\mathbf{Y}}$ for $t \leq s \leq T.$ By uniqueness of the solution, we have $$\Psi_{t+h \leftarrow 0}^{\mathbf{Z^{\mathbf{Y}}}|_{[0,t+h]}}(x_0) = \Psi_{t+h \leftarrow t}^{\mathbf{Z^{\mathbf{Y}}}|_{[t,t+h]}}\circ \Psi_{t \leftarrow 0}^{\mathbf{Z^{\mathbf{Y}}}|_{[0,t]}}(x_0).$$ Since $\mathbf{Z}^{\mathbf{Y}}|_{[t,t+h]}$ is $(\mathcal{F}^B_{t,t+h})-$measurable, it follows that $\Psi_{t+h \leftarrow t}^{\mathbf{Z^{\mathbf{Y}}}|_{[t,t+h]}}(\cdot)$ is independent of $\mathcal{F}_t^B$, and we can conclude $$E\left [k(X^{\mathbf{Y}}_{t+h})|\mathcal{F}_t^B\right] =\left.  E\left [k(\Psi_{t+h \leftarrow t}^{\mathbf{Z^{\mathbf{Y}}}|_{[t,t+h]}}(y))\right ] \right | _{y=\Psi_{t \leftarrow 0}^{\mathbf{Z^{\mathbf{Y}}}|_{[0,t]}}(x_0)}.$$ But this readily implies \eqref{eq:MarkovPropertyAppendix}, since the left hand side is a deterministic function of $\Psi_{t \leftarrow 0}^{\mathbf{Z^{\mathbf{Y}}}|_{[0,t]}}(x_0)=X^{\mathbf{Y}}_t$. Next,  by continuity of the It{\^o}-Lyons map, it follows that $X_s^{t,
  x, \tmmathbf{I}}$ is the continuous image of the joint rough path
  $\tmmathbf{Z}^{\tmmathbf{I}}$, and hence measurable with respect to $\sigma
  (\{ \tmmathbf{Z}^{\mathbf{I}}_u : u \leq s \}) \subseteq \mathcal{F}_s$, and thus adapted. In Remark \ref{CorresMG} we saw that $[\mathbf{I}] = [I]=\int V_tdt$ almost surely, and therefore we have $\mathbf{v}_t^{\mathbf{I}}:= \sqrt{\frac{d[\mathbf{I}]_t}{dt}} =\sqrt{V_t}$, which is adapted since the map $\mathbb{S}_+^{d\times d} \ni A \mapsto \sqrt{A} \in \mathbb{S}_+^{d\times d}$ is continuous, and hence \begin{equation*}
      Z_t=\begin{pmatrix}
    M^{I}_t \\ I_t
\end{pmatrix}= \begin{pmatrix}
    \int_0^t\sqrt{V_t}dB_s \\ \int_0^tv_sdW_s
\end{pmatrix} \in \mathbb{R}^{2d}. 
  \end{equation*} By \cite[Lemma 7.3]{friz2020course} it follows that the following defines a controlled rough path
  \[ A = (g, f) (s, X_s^{t, x, \tmmathbf{I}}), \quad A' = (D (g, f)
     (s, X_s^{t, x, \tmmathbf{I}})) (g, f) (s, X_s^{t, x, \tmmathbf{I}}).
  \]
  Moreover, $(A, A')$ is adapted and continuous, and therefore we have almost surely
  \[ X_s^{t, x, \tmmathbf{I} (\omega)} = x + \left. \int_t^s A_u (\omega) d
     \tmmathbf{Z}_u^{\mathbf{Y}}  (\omega) \right|_{\tmmathbf{Y} = \tmmathbf{I} (\omega)}
     \nobracket = x + \left (\int_t^s A_u dZ_u\right )(\omega), \]
  where we used Lemma \ref{roughIto} and Lemma \ref{consistenyALemma}, and the fact that $[M^{\mathbf{I}}, I] = 0$ for
  the last equality. This shows that $X^{t,x,\mathbf{I}}$ is indistinguishable from the unique solution to the SDE \begin{equation*}
      d\hat{X}^{t,x}_t=x, \quad d\hat{X}^{t,x}_s = f(s,\hat{X}_s^{t,x})v_sdW_s+g(s,\hat{X}_s^{t,x})\sqrt{V_s}dB_s, \quad t < s \leq T.
  \end{equation*} But by independence of $W$ and $B$, one can notice that the conditional distributions of $\int_t^{\cdot} v_sdB_s$ and $\int_t^{\cdot} \sqrt{V_s}dB_s$, given $\mathcal{F}^W_T\lor \mathcal{F}^B_t$, are Gaussian with covariance matrix  $\int_t^{\cdot} V_sds$, and thus coincide. In particular, it holds that \begin{equation*}
      \mathrm{Law}\left (\left.\int_t^{\cdot}v_sdW_s,\int_{t}^{\cdot} \sqrt{V_s}dB_s \right| \mathcal{F}^W_T\lor \mathcal{F}^B_t \right ) = \mathrm{Law}\left (\left. \int_t^{\cdot} v_sdW_s,\int_{t}^{\cdot} v_sdB_s \right | \mathcal{F}^W_T\lor \mathcal{F}^B_t \right ).
  \end{equation*} 
  By strong uniqueness of the SDE \eqref{dynamicsA}, see \cite[Chapter 3 Theorem 7]{protter2005stochastic}, we can conclude that \begin{equation*}
      \mathrm{Law}(X^{t,x,\mathbf{I}}| \mathcal{F}^W_T\lor \mathcal{F}^B_t) =\mathrm{Law}(\hat{X}^{t,x}| \mathcal{F}^W_T\lor \mathcal{F}^B_t) = \mathrm{Law}(X^{t,x}| \mathcal{F}^W_T\lor \mathcal{F}_t^B ).
  \end{equation*} Moreover, if we additionally assume that $v\in \mathbb{S}^d_+$, then clearly $X^{t,x,\mathbf{I}}=\hat{X}^{t,x}=X^{t,x}$ almost surely. Finally, we are left with the equality $\mathrm{Law}\left (\left. X^{t,x,\mathbf{I}}\right |\mathcal{F}^W_T \lor \mathcal{F}_t^B\right)(\omega) =\mathrm{Law}(\left. X^{t,x,\mathbf{Y}})\right |_{\mathbf{Y=\mathbf{I}(\omega)}}.$ By continuity of $\mathbf{Y} \mapsto X^{t,x,\mathbf{Y}}$, it follows that the latter is a stochastic kernel. Moreover, for all measurable sets $A$, we can use the continuity together with the $(\mathcal{F}^B_t)-$Markov property to see that almost surely $$P\left (X^{t,x,\mathbf{I}}\in A  | \mathcal{F}_T^{W}\lor \mathcal{F}_t^B  \right )(\omega) = P(X^{t,x,\mathbf{Y}}\in A)|_{\mathbf{Y}=\mathbf{I}(\omega)},$$ and thus by uniqueness of the regular conditional distribution the claim follows.
\end{proof} 
\begin{remark}\label{remarkhybridinter}
     Finally, notice that for the geometric rough path $\mathbf{Y}^g$, we have $\mathbb{Y}^g = \mathbb{Y} + \frac{1}{2}\delta [\mathbf{Y}]$, and hence\begin{align*}
  \mathbb{Z}^{\mathbf{Y}}_{s,t}  \assign \left(\begin{array}{cc}
    \int_s^t M_{s,r}\otimes dM_r & \int_s^t M_{s,r}\otimes dY_r\\
    \int_s^t Y_{s,r}\otimes dM_r & \mathbb{Y}_{s,t} 
  \end{array}\right) & = \left(\begin{array}{cc}
    \int_s^t M_{s,r}\otimes dM_r & \int_s^t M_{s,r}\otimes dY_r\\
    \int_s^t Y_{s,r}\otimes dM_r & \mathbb{Y}^g_{s,t} - \frac{1}{2}\delta[\mathbf{Y}]_{s,t}
  \end{array}\right) \\ &= \mathbb{Z}^{\mathbf{Y}^g}_{s,t}+\left(\begin{array}{cc}
    0 & 0\\
    0 & -\frac{1}{2}\delta[\mathbf{Y}]_{s,t} 
  \end{array}\right).
\end{align*} Now by standard properties of rough integration, see \cite[Chapter 4 Example 4.14]{friz2020course}, it follows that \begin{equation*}
    dX_s^{t,x,\mathbf{Y}} = (g,f)(s,X_s^{t,x,\mathbf{Y}})d\mathbf{Z}^{\mathbf{Y}}_s = (g,f)(s,X_s^{t,x,\mathbf{Y}})d\mathbf{Z}^{\mathbf{Y}^g}_s + f_0(s,X_s^{t,x,\mathbf{Y}})ds,
\end{equation*} where $f_0:[0,T]\times \mathbb{R}^{d_X}\rightarrow \mathbb{R}^{d_X}$ is such that $f_0^i(t,x)= -\frac{1}{2}\sum_{j,k,l}\partial_{x_k}f^{ij}(t,x)f^{kl}(t,x)(\mathbf{V}_t^{\mathbf{Y}})^{lj}$ for $i=1,\dots,d_X$.
\end{remark}